\documentclass{article}
\usepackage{amsfonts, amsmath, amsthm, latexsym, amscd, graphicx, xypic, psfrag}
\usepackage{epic, eepic, epsfig, setspace, ifthen, amssymb, xy, yhmath, lscape}
\usepackage{bbold, amscd, paralist, graphics}
\usepackage{scalerel,stackengine}

\usepackage[colorlinks=true]{hyperref}

\def\<#1,#2>{\langle #1,#2 \rangle}

\def\bothID{\rlap{\hbox to.97\wd0{\hss\vrule height.06\ht0 width.82\wd0}}
\copy0\rlap{\kern-.36\wd0\vrule height1.05\ht0 width.05\ht0}\kern.14\wd0}

\newtheorem{theorem}{Theorem}

\newtheorem{corollary}[theorem]{Corollary}

\newtheorem{definition}[theorem]{Definition}

\newtheorem{lemma}[theorem]{Lemma}

\newtheorem{proposition}[theorem]{Proposition}
\newtheorem{remark}[theorem]{Remark}

\stackMath
\newcommand\reallywidehat[1]{%
\savestack{\tmpbox}{\stretchto{%
  \scaleto{%
    \scalerel*[\widthof{\ensuremath{#1}}]{\kern-.6pt\bigwedge\kern-.6pt}%
    {\rule[-\textheight/2]{1ex}{\textheight}}
  }{\textheight}%
}{0.5ex}}%
\stackon[1pt]{#1}{\tmpbox}%
}
\begin{document}

\title{The Black-Scholes Equation in Presence of Arbitrage}

\author{Simone Farinelli\\
        Core Dynamics GmbH\\
        Scheuchzerstrasse 43\\
        CH-8006 Zurich\\
        Email: simone@coredynamics.ch\\and\\
        Hideyuki Takada\\
        Department of Information Science\\
        Narashino Campus, Toho University\\
        2-2-1-Miyama, Funabashi-Shi\\ J-274-8510 Chiba\\
        Email: hideyuki.takada@is.sci.toho-u.ac.jp
        }

\maketitle

\begin{abstract}
We apply Geometric Arbitrage Theory to obtain results in Mathematical Finance, which do not need stochastic differential geometry in their formulation. First, for generic market dynamics given by a subclass of multidimensional It\^{o} processes we specify and prove the equivalence between No-Free-Lunch-with-Vanishing-Risk (NFLVR) and expected utility maximization. As a by-product we provide a geometric characterization of the No-Unbounded-Profit-with-Bounded-Risk (NUPBR) condition given by the zero curvature (ZC) condition for this subclass of It\^{o} processes. Finally, we extend the Black-Scholes partial differential equation to markets allowing arbitrage.\\\\
\text{Keywords: NFLVR, NUPBR, Geometric Arbitrage Theory, Non linear Black Scholes PDE}\\\\
\text{AMS: 91G80, 53C07}
\end{abstract}

\tableofcontents

\section{Introduction}
This paper provides applications of a conceptual structure - called Geometric Arbitrage Theory (GAT in short) - to prove results in financial mathematics which are comprehensible without the use of stochastic differential geometry and extend well known classical facts. We expect therefore to make GAT accessible to a wider public in the mathematical finance community.\par
GAT rephrases classical stochastic finance in stochastic differential geometric terms in order to characterize arbitrage. The main idea of the GAT approach consists of modeling markets made of basic financial instruments together with their term structures as principal fibre bundles. Financial features of this market - like no arbitrage and equilibrium - are then characterized in terms of standard differential geometric constructions - like curvature - associated to a natural connection in this fibre bundle. Principal fibre bundle theory has been heavily exploited in theoretical physics as the language in which laws of nature can be best formulated by providing an invariant framework to describe physical systems and their dynamics. These ideas can be carried over to mathematical finance and economics. A market is a financial-economic system that can be described by an appropriate principle fibre bundle. A principle like the invariance of market laws under change of num\'{e}raire can be seen then as gauge invariance.\par
The fact that gauge theories are the natural language to describe economics was first proposed by Malaney and Weinstein in the context of the economic index problem (\cite{Ma96}, \cite{We06}). Ilinski (see \cite{Il00} and \cite{Il01}) and Young \cite{Yo99} proposed to view arbitrage as the curvature of a gauge connection, in analogy to some physical theories. Independently, \cite{SmSp98} further developed \cite{FlHu96} seminal work and utilized techniques from differential geometry to reduce the complexity of asset models before stochastic modeling.\par
Why is arbitrage modelling important? The no arbitrage condition is only an approximation and it is not fulfilled when we consider real markets. This is the case for non traded assets, traded assets when the frequency of the trades falls below $2$ minutes (cf. \cite{FaVa12}) or electricity markets, where we do not have the possibility of completely liquidating the portfolio at any given time, as we implicitly assume in mathematical finance. This has been recognized for a long time and in recent years the modelling of markets allowing for arbitrage beyond pathological cases has made a relevant progress (see f.i. \cite{HuPr15, Ru13}). The benchmark approach to mathematical finance models markets by \cite{HePl06} allowing for arbitrage even if this is not explicitly mentioned.\par
This paper is structured as follows. Section 2 reviews classical stochastic finance and  Geometric Arbitrage Theory, summarizing \cite{Fa15}, where GAT has been given a rigorous mathematical foundation utilizing the formal background of stochastic differential geometry as in \cite{Schw80}, \cite{El82}, \cite{Em89}, \cite{HaTh94},
\cite{St00} and \cite{Hs02}. Arbitrage is seen as curvature of a principal fibre bundle representing the market which defines the quantity of arbitrage associated to it. The zero curvature condition is a weaker condition than No-Free-Lunch-with-Vanishing-Risk (NFLVR). It becomes equivalent under additional assumptions introduced for a guiding example, a market whose asset prices are It\^{o} processes. In general, the zero curvature condition follows from the No-Unbounded-Profit-with-Bounded-Risk (NUPBR) condition, as we prove in Section 3, where we analyze the relationship between arbitrage and expected utility maximization. The equivalence is proved for a certain subclass of It\^{o} processes. In Section 4, GAT is applied to prove an extension of the Black Scholes PDE in the case of markets allowing for arbitrage. Section 5 concludes, and Appendix A reviews Nelson's stochastic derivatives.

%
%
%
%
\section{Geometric Arbitrage Theory Background}\label{section2}
In this section we explain the main concepts of Geometric Arbitrage Theory introduced in \cite{Fa15}, to which we refer for proofs and additional examples.
Since the differential geometric thinking is not so widespread in the mathematical finance community, we explain in detail the reformulation of the asset model as principal fibre bundle with a connection, whose curvature can be seen as a measure of arbitrage. New results and more pedagogical results in comparison to \cite{Fa15} are provided.

\subsection{The Classical Market Model}\label{StochasticPrelude}

In this subsection we will summarize the classical set up, which will be rephrased in Section \ref{foundations} in differential geometric terms. We basically follow \cite{HuKe04} and the ultimate reference \cite{DeSc08}.\par We assume continuous time trading and that the set of trading dates is $[0,+\infty[$. This assumption is general enough to embed the cases of finite and infinite discrete times as well as the one with a finite horizon in continuous time. Note that while it is true that in the real world trading occurs at discrete times only, these are not known a priori and can be virtually any points in the time continuum. This motivates the technical effort of continuous time stochastic finance.\par The uncertainty is modelled by a filtered probability space $(\Omega,\mathcal{A}, \mathbb{P})$, where $\mathbb{P}$ is the statistical (physical) probability measure, $\mathcal{A}=\{\mathcal{A}_t\}_{t\in[0,+\infty[}$ an increasing family of sub-$\sigma$-algebras of $\mathcal{A}_{\infty}$ and $(\Omega,\mathcal{A}_{\infty}, \mathbb{P})$ is a probability space. The filtration $\mathcal{A}$ is assumed to satisfy the usual conditions, that is
\begin{itemize}
\item right continuity: $\mathcal{A}_t=\bigcap_{s>t}\mathcal{A}_s$ for all $t\in[0,+\infty[$.
\item $\mathcal{A}_0$ contains all null sets of
$\mathcal{A}_{\infty}$.
\end{itemize}

The market consists of finitely many \textbf{assets} indexed by $j=1,\dots,N$, whose \textbf{nominal prices} are given by the vector valued semimartingale $S:[0,+\infty[\times\Omega\rightarrow\mathbb{R}^N$ denoted by $(S_t)_{t\in[0,+\infty[}$ adapted to the filtration $\mathcal{A}$.
The stochastic process $(S^ j_t)_{t\in[0,+\infty[}$ describes the price at time $t$ of the $j$-th asset in terms of  unit of cash \textit{at time $t=0$}. More precisely, we assume the existence of a $0$-th asset, the \textbf{cash}, a strictly positive semimartingale, which evolves according to $S_t^0=\exp(\int_0^tdu\,r^0_u)$, where the integrable semimartingale $(r^0_t)_{t\in[0,+\infty[}$ represents the continuous interest rate provided by the cash account: one always knows in advance what the interest rate on the own bank account is, but this can change from time to time. The cash account is therefore considered the locally risk less asset in contrast to the other assets, the risky ones. In the following we will mainly utilize \textbf{discounted prices}, defined as $\hat{S}_t^j:=S_t^j/S^{0}_t$, representing the asset prices in terms of \textit{current} unit of cash.\par
 We remark that there is no need to assume that asset prices are positive. But, there must be at least one strictly positive asset, in our case the cash. If we want to renormalize the prices by choosing another asset instead of the cash as reference, i.e. by making it to our \textbf{num\'{e}raire}, then this asset must have a strictly positive price process. More precisely, a generic num\'{e}raire is an asset, whose nominal price is represented by a strictly positive stochastic process $(B_t)_{t\in[0,+\infty[}$, and which is a portfolio of the original assets $j=0,1,2,\dots,N$. The discounted prices of the original assets are  then represented in terms of the num\'{e}raire by the semimartingales $\hat{S}_t^j:=S_t^j/B_t$.\par We assume that there are no transaction costs and that short sales are allowed. Remark that the absence of transaction costs can be a serious limitation for a realistic model. The filtration $\mathcal{A}$ is not necessarily generated by the price process $(S_t)_{t\in[0,+\infty[}$: other sources of information than prices are allowed. All agents have access to the same information structure, that is to the filtration $\mathcal{A}$.\par

Let $v$ be a positive real number. A $v$-admissible \textbf{strategy} $x=(x_t)_{t\in[0,+\infty[}$ is a $S$-integrable predictable process for which the It\^{o} integral $\int_0^tx\cdot dS\ge-v$ a.s. for all $t\ge0$ with $x_0=0$. A strategy is admissible if it is $v$-admissible for some $v\ge0$.

\begin{definition}[\textbf{Arbitrage}] Let the process $(S_t)_{[0,+\infty[}$ be a semimartingale and $(x_t)_{t\in[0,+\infty[}$ be admissible self-financing strategy. Let us consider trading up to time $T\le\infty$. The portfolio wealth at time $t$ is given by $V_{t}(x):=V_0+\int_0^tx_u\cdot dS_u$, and we denote by $K_0$ the subset of $L^0(\Omega, \mathcal{A}_{T},\mathbb{P})$ containing all such $V_T(x)$, where $x$ is any admissible self-financing strategy.
We define
\begin{itemize}
\item $C_0:=K_0-L_+^0(\Omega, \mathcal{A}_{T},\mathbb{P})$.
\item $C:=C_0\cap L^{\infty}(\Omega, \mathcal{A}_{T},\mathbb{P})$.
\item $\bar{C}$: the closure of $C$ in $L^{\infty}$ with respect to the norm topology.
\item $\mathcal{V}^{V_0}:=\left\{(V_{t})_{t\in[0,+\infty[}\,\big{|}\, V_t=V_t(x), \,\text{where } x \text{ is } V_0\text{-admissible} \right\}$.
\item $\mathcal{V}_T^{V_0}:=\left\{V_T\,\big{|}\,(V_{t})_{t\in[0,+\infty[}\in\mathcal{V}^{V_0}\right\}$:  terminal wealth for $V_0$-admissible self-financing strategies.
\end{itemize}
And let $L_+^{\infty}(\Omega, \mathcal{A}_{T},\mathbb{P})$ be the set of positive random variables in $L^{\infty}(\Omega, \mathcal{A}_{T},\mathbb{P})$.
We say that $S$ satisfies
\begin{itemize}
\item \textbf{(NA), no arbitrage}, if and only if $C \cap L_+^{\infty}(\Omega, \mathcal{A}_{T},\mathbb{P})=\{0\}$.
\item \textbf{(NFLVR), no-free-lunch-with-vanishing-risk},  if and only if $\bar{C} \cap L_+^{\infty}(\Omega, \mathcal{A}_{T},\mathbb{P})=\{0\}$.
\item \textbf{(NUPBR), no-unbounded-profit-with-bounded-risk}, if and only if $\mathcal{V}_T^{V_0}$ is bounded in $L^0$ for some $V_0>0$.
\end{itemize}
\end{definition}

\noindent The relationship between these three different types of arbitrage has been elucidated in \cite{DeSc94} and in \cite{Ka97} with the proof of the following result.
\begin{theorem}
\begin{equation*}
(NFLVR) \Leftrightarrow (NA) + (NUPBR).
\end{equation*}
\end{theorem}
\begin{remark} We recall that, as shown in \cite{DeSc94, Ka97, KaKr94, KaKa07}, (NUPBR) is equivalent to (NAA1), i.e.
no asymptotic arbitrage of the $1$st kind , and equivalent to (NA1), i.e. no arbitrage of the 1st kind.
\end{remark}

\subsection{Geometric Reformulation of the Market Model: Primitives}
We are going to introduce a more general representation of the market model introduced in Section \ref{StochasticPrelude}, which better suits to the arbitrage modeling task.
\begin{definition}\label{defi1}
A \textbf{gauge} is an ordered pair of two $\mathcal{A}$-adapted real valued semimartingales $(D, P)$, where $D=(D_t)_{t\ge0}:[0,+\infty[\times\Omega\rightarrow\mathbb{R}$ is called \textbf{deflator} and $P=(P_{t,s})_{t,s}:\mathcal{T}\times\Omega\rightarrow\mathbb{R}$, which is called \textbf{term structure}, is considered as a stochastic process with respect to the time $t$, termed \textbf{valuation date} and $\mathcal{T}:=\{(t,s)\in[0,+\infty[^2\,|\,s\ge t\}$. The parameter $s\ge t$ is referred as \textbf{maturity date}. The following properties must be satisfied almost surely for all $t, s$ such that $s\ge t\ge 0$; $P_{t,s}>0, P_{t,t}=1$.
\end{definition}

Deflators and term structures can be considered \textit{outside the context of fixed income.} An arbitrary financial instrument is mapped to a gauge $(D, P)$ with the following economic interpretation:
\begin{itemize}
\item Deflator: $D_t$ is the value of the financial instrument at time $t$ expressed in terms of some num\'{e}raire. If we choose the cash account, the $0$-th asset as num\'{e}raire, then we can set $D_t^j:=\hat{S}_t^j=\frac{S_t^j}{S_t^0}\quad(j=1,\dots N)$.
\item Term structure: $P_{t,s}$ is the value at time $t$ (expressed in units of deflator at time $t$) of a synthetic zero coupon bond with maturity $s$ delivering one unit of financial instrument at time $s$. It represents a term structure of forward prices with respect to the chosen num\'{e}raire.
\end{itemize}
\noindent We point out that there is no unique choice for deflators and term structures describing an asset model. For example, if a set of deflators qualifies, then we can multiply every deflator  by the same positive semimartingale to obtain another suitable set of deflators. Of course term structures have to be modified accordingly. The term ``deflator" is clearly inspired by actuarial mathematics and  was first introduced in \cite{SmSp98}. In the present context it refers to an asset value up division by a strictly positive semimartingale (which can be the state price deflator if this exists and it is made to the num\'{e}raire). There is no need to assume that a deflator is a positive process. However, if we want to make an asset to our num\'{e}raire, then we have to make sure that the corresponding deflator is a strictly positive stochastic process.

\subsection{Geometric Reformulation of the Market Model: Portfolios}\label{trans}
We want now to introduce transforms of deflators and term structures in order to group gauges containing the same (or less) stochastic information. That for, we will consider \textit{deterministic} linear combinations of assets modelled by the same gauge (e. g. zero bonds of the same credit quality with different maturities).

\begin{definition}\label{gaugeTransforms2}
Let $\pi:[0, +\infty[\longrightarrow \mathbb{R}$ be a deterministic cashflow intensity (possibly generalized) function. It induces a \textbf{gauge transform} $(D,P)\mapsto \pi(D,P):=(D,P)^{\pi}:=(D^{\pi}, P^{\pi})$ by the formulae
\begin{equation*}
 D_t^{\pi}:=D_t\int_0^{+\infty}dh\,\pi_h P_{t, t+h},\qquad
P_{t,s}^{\pi}:=\frac{\displaystyle \int_0^{+\infty}dh\,\pi_h P_{t, s+h}}{\displaystyle  \int_0^{+\infty}dh\,\pi_h P_{t, t+h}}.
\end{equation*}
\end{definition}
\begin{remark} The cashflow intensity $\pi$ specifies the bond cashflow
structure. The bond value at time $t$ expressed in terms of the
market model  num\'{e}raire is given by $D_t^{\pi}$.  The term
structure of forward prices for the bond future expressed in terms
of the bond current value is given by $P_{t,s}^{\pi}$.
\end{remark}
\begin{proposition}\label{convRef}
Gauge transforms induced by cashflow vectors have the following property:
\begin{equation}
((D,P)^{\pi})^{\nu}= ((D,P)^{\nu})^{\pi} = (D,P)^{\pi\ast\nu},
\label{comm}
\end{equation}
where $\ast$ denotes the convolution product of two cashflow vectors or intensities respectively:
\begin{equation}\label{convdef}
    (\pi\ast\nu)_t:=\int_0^tdh\,\pi_h\nu_{t-h}.
\end{equation}
\end{proposition}
\begin{proof}
We can observe that
\begin{align*}
(D_t^{\pi})^{\nu} &= D_t^{\pi} \int_0^{+\infty}dh\, \nu_{h} P_{t,t+h}^{\pi} = D_t \int_0^{+\infty}dh\, \nu_h \int_0^{+\infty}du\,\pi_u P_{t,t+h+u}.
\end{align*}
By changing variables $v:=h+u$, one has
\begin{align*}
(D_t^{\pi})^{\nu}   &= D_t \int_0^{+\infty}dv  \Bigl( \int_0^v dh\, \nu_h \pi_{v-h} \Bigl)   P_{t,t+v} =  (D_t)^{\pi\ast\nu}
\end{align*}
and this coincide with $(D_t^{\nu})^{\pi}$, proving the first component of (\ref{comm}). The second component can be derived similarly.
\end{proof}

The convolution of two non-invertible gauge transform is non-invertible. The convolution of a non-invertible with an invertible gauge transform is non-invertible.

\begin{definition}\label{int}
The term structure can be written as a functional of the instantaneous forward rate $f$ defined as
\begin{equation*}
  f_{t,s}:=-\frac{\partial}{\partial s}\log P_{t,s},\quad
  P_{t,s}=\exp\left(-\int_t^sdhf_{t,h}\right),
\end{equation*}
\noindent and
\begin{equation}
 r_t:=\lim_{s\rightarrow t^+}f_{t,s}
\end{equation}
\noindent is termed short rate.
\end{definition}
\begin{remark}
Since $(P_{t,s})_{t,s}$ is a $t$-stochastic process (semimartingale) depending on a parameter $s\ge t$, the $s$-derivative can be defined deterministically, and the expressions above make sense pathwise in a both classical and generalized sense. In a generalized sense we will always have a $\mathcal{D}^{\prime}$ derivative for any $\omega\in \Omega$; this corresponds to a classic $s$-continuous  derivative if $P_{t,s}(\omega)$ is a $C^1$-function of $s$ for any fixed $t\ge0$ and $\omega\in\Omega$.
\end{remark}
\begin{remark} The special choice of vanishing interest rate $r\equiv0$ or flat term structure $P\equiv1$ for all assets corresponds to the classical model, where only asset prices and their dynamics are relevant.
\end{remark}

\subsection{Arbitrage Theory in a Differential Geometric Framework}\label{foundations} Now we are in the position to rephrase the asset model presented in Subsection \ref{StochasticPrelude} in terms of a natural geometric language. Given $N$ base assets we want to construct a portfolio theory and study arbitrage and thus we cannot a priori assume the existence of a risk neutral measure or of a state price deflator. In terms of differential geometry, we will adopt the mathematician's and not the physicist's approach. The market model is seen as a principal fibre bundle of the (deflator, term structure) pairs, discounting and portfolio rebalance (or foreign exchange) as a parallel transport, num\'{e}raire as global section of the gauge bundle, arbitrage as curvature.  The no-unbounded-profit-with-bounded-risk condition is proved to imply a zero curvature condition.

\subsubsection{Market Model as Principal Fibre Bundle}
Let us consider -in continuous time- a market with $N$ assets and a num\'{e}raire. A general portfolio at time $t$ is described by the vector of nominals $x\in X$, for an open set $X\subset\mathbb{R}^N$. By nominals $x^1,\dots,x^N$ we mean the number of assets that we hold in our portfolio. Following Definition \ref{defi1}, the asset model consisting in $N$ synthetic zero bonds is described by means of the gauges
\begin{equation*}(D^j,P^j)=((D_t^j)_{t\in[0, +\infty[},(P_{t,s}^j)_{s\ge t}),\end{equation*}
\noindent where $D^j$ denotes the deflator and $P^j$ the term structure for $j=1,\dots,N$.
More exactly: $D_t^j$ is the value of the $j$-th financial instrument at time $t$ expressed in terms of some num\'{e}raire, and
$P_{t,s}^j$ is the value at time $t$ (expressed in units of deflator $D_t^j$ at time $t$ ) of the $j$-th synthetic zero coupon bond with maturity $s$ delivering one unit of financial instrument at time $s$. \\
The term structure  can be written as
\begin{equation*}P_{t,s}^j=\exp\left(-\int_t^sf^j_{t,u}du\right),\end{equation*}
where $f^j$ is the instantaneous forward rate process for the $j$-th asset and the corresponding short rate is given by $r_t^j:=\lim_{u\rightarrow t^+}f^j_{t,u}$. For a portfolio with nominals $x\in X\subset\mathbb{R}^N$ we define
\begin{equation*}
D_t^x:=\sum_{j=1}^Nx_jD_t^j\quad
f_{t,u}^x:=\sum_{j=1}^N\frac{x_jD_t^j}{\sum_{k=1}^Nx_kD_t^k}f_{t,u}^j\quad
P_{t,s}^x:=\exp\left(-\int_t^sf^x_{t,u}du\right).
\end{equation*}
The short rate writes
\begin{equation*}
r_t^x:=\lim_{u\rightarrow t^+}f^x_{t,u}=\sum_{j=1}^N\frac{x_jD_t^j}{\sum_{k=1}^Nx_kD_t^k}r_t^j.
\end{equation*}
The image space of all possible strategies reads
\begin{equation*}M:=\{(x,t)\in X\times[0, +\infty[\}.\end{equation*}
In Subsection \ref{trans} cashflow intensities and the corresponding gauge transforms were introduced. They have the structure of an Abelian semigroup
\begin{equation*}
 H:=\mathcal{E}^{\prime}([0, +\infty[,\mathbb{R})=\{F\in\mathcal{D}^{\prime}([0,+\infty[)\mid \text{supp}(F)\subset[0, +\infty[\text{ is compact}\},
\end{equation*}
where the semigroup operation on distributions with compact support is the convolution (see \cite{Ho03}, Chapter IV), which extends the convolution of regular functions as defined by formula (\ref{convdef}).
\begin{definition}
The \textbf{Market Fibre Bundle} is defined as the fibre bundle of gauges
\begin{equation*}
\mathcal{B}:=\{ ({D^x_t},{P^x_{t,\,\cdot}})^{\pi }|\,(x,t)\in M, \pi\in G\}.
\end{equation*}
\end{definition}
\noindent The cashflow intensities defining invertible transforms constitute an Abelian group
\begin{equation*}
G:=\{\pi\in H |\text{ it exists } \nu\in H\text{ such that }\pi\ast\nu=\delta \}\subset \mathcal{E}^{\prime}([0, +\infty[,\mathbb{R}).
\end{equation*}
where $\delta $ is Dirac delta function acts as an identity element. From Proposition \ref{convRef} we obtain
\begin{theorem} The market fibre bundle $\mathcal{B}$ has the structure of a $G$-principal fibre bundle given by the action
\begin{equation*}
\begin{split}
\mathcal{B}\times G &\longrightarrow\mathcal{B}\\
 ((D,P), \pi)&\mapsto (D,P)^{\pi}=(D^{\pi},P^{\pi})
\end{split}
\end{equation*}
\noindent The group $G$ acts freely and differentiably on $\mathcal{B}$ to the right.
\end{theorem}
The market fibre bundle repackages all the information concerning market dynamics of the asset futures and their underlyings. The principal bundle structure reflects the portfolio construction possibilities at a fixed time, as well as the synthetic bond construction possibilities for given cash flow patterns specified by the gauge transforms.

\subsubsection{Nelson Weak $\mathcal{D}$-Differentiable Market Model} We continue to reformulate the classic asset model introduced in Subsection \ref{StochasticPrelude} in terms of stochastic differential geometry.

\begin{definition}\label{weakMM}
 A \textbf{Nelson weak $\mathcal{D}$-differentiable market model} for $N$ assets is described by $N$ gauges which are Nelson weak $\mathcal{D}$-differentiable with respect to the time variable. More exactly, for all $t\in[0,+\infty[$ and $s\ge t$ there is an open time interval $I\ni t$ such that for the deflators $D_t:=[D_t^1,\dots,D_t^N]^{\dagger}$ and the term structures $P_{t,s}:=[P_{t,s}^1,\dots,P_{t,s}^N]^{\dagger}$, the latter seen as processes in $t$ and parameter $s$, there exist a weak $\mathcal{D}$-derivative with respect to the time variable $t$ (see Appendix \ref{Derivatives}). The short rates are defined by $r_t:=\lim_{s\rightarrow t^{+}}\frac{\partial}{\partial s}\log P_{t,s}$.\par
A strategy is a curve $\gamma:I\rightarrow X$ in the portfolio space parameterized by the time. This means that the allocation at time $t$ is given by the vector of nominals $x_t:=\gamma(t)$. We denote by $\bar{\gamma}$ the lift of $\gamma$ to $M$, that is $\bar{\gamma}(t):=(\gamma(t),t)$. A strategy is said to be \textbf{closed} if it represented by a closed curve.  A \textbf{weak $\mathcal{D}$-admissible strategy} is predictable and weak $\mathcal{D}$-differentiable.
\end{definition}
\begin{remark}
We require weak $\mathcal{D}$-differentiability and not strong $\mathcal{D}$-differentiability because imposing a priori regularity properties on the trading strategies corresponds to restricting the class of admissible strategies with respect to the classical notion of Delbaen and Schachermayer.
Every (no-)arbitrage consideration depends crucially on the chosen definition of admissibility.
Therefore, restricting the class of admissible strategies may lead to the automatic exclusion of potential arbitrage opportunities, leading to vacuous statements for kinds of
Fundamental Theorem of Asset Pricing.
An admissible strategy in the classic sense (see Section \ref{section2}) is weak $\mathcal{D}$-differentiable.
\end{remark}
\noindent In general the allocation can depend on the state of the nature i.e. $x_t=x_t(\omega)$ for $\omega\in\Omega$.
\begin{proposition}
A weak $\mathcal{D}$-admissible strategy is self-financing if and only if
\begin{equation}\label{sf}
\mathcal{D}(x_t\cdot D_t)=x_t\cdot \mathcal{D}D_t-\frac{1}{2}\mathfrak{D}_*\left<x,D\right>_t\text{ or }
\mathcal{D}x_t\cdot D_t=-\frac{1}{2}\mathfrak{D}_*\left<x,D\right>_t\text{ or }
\mathfrak{D}x_t\cdot D_t=0,
\end{equation}
almost surely. The bracket $\left<\cdot,\cdot\right>$ denotes the continuous part of the quadratic covariation.
\end{proposition}
\begin{proof}
The strategy is self-financing if and only if
\begin{equation*}
x_t\cdot D_t =x_0 \cdot D_0 + \int_0^t x_u \cdot dD_u,
\end{equation*}
which is, symbolizing It\^{o}'s differential $d$, equivalent to
\begin{equation}\label{53}
\mathfrak{D}(x_t\cdot D_t) = x_t\cdot \mathfrak{D}D_t,
\end{equation}
or equivalent to
\begin{equation}\label{54}
\mathfrak{D}x_t\cdot D_t = 0.
\end{equation}

The self-financing condition can be expressed by means of the anticipative differential $d_*$ as
\begin{equation*}
x_t\cdot D_t = x_0\cdot D_0 + \int_0^t x_u\cdot d_* D_u - \int_0^t d\left<x,D\right>_u,
\end{equation*}
which is equivalent to
\begin{equation}
\mathfrak{D}_* (x_t\cdot D_t) = x_t\cdot \mathfrak{D}_* D_t - \mathfrak{D}_* \left<x,D\right>_t.
\label{55}
\end{equation}
By summing equations (\ref{53}) and (\ref{55}) we obtain
\begin{equation*}
\mathcal{D} (x_t\cdot D_t)=\frac{1}{2}(\mathfrak{D}+\mathfrak{D}_*)(x_t\cdot D_t) = x_t\cdot \mathcal{D}D_t - \frac{1}{2}\mathfrak{D}_* \left<x,D\right>_t.
\end{equation*}
To prove the second statement in expression (\ref{sf}) we consider the integration by parts formula for It\^{o}'s integral
\begin{equation*}
\int_0^t x_u\cdot dD_u + \int_0^t D_u\cdot dx_u = x_t \cdot D_t - x_0\cdot D_0 - \left<x,D\right>_t,
\end{equation*}
which, expressed in terms of Stratonovich's integral, leads to
\begin{equation*}
\int_0^t x_u \circ dD_u - \frac{1}{2} \left<x,D\right>_t + \int_0^t D_u\circ dx_u -\frac{1}{2}\left<x,D\right>_t = x_t\cdot D_t - x_0\cdot D_0 -\left<x,D\right>_t.
\end{equation*}
By taking Stratonovich's derivative $\mathcal{D}$ on both side we get
\begin{equation*}
\mathcal{D}(x_t\cdot D_t) = \mathcal{D}x_t\cdot D_t + x_t\cdot \mathcal{D}D_t,
\end{equation*}
which, together with the first statement in expression (\ref{sf}) proves the second one.
\end{proof}

For the reminder of this paper unless otherwise stated we will deal only with weak $\mathcal{D}$-differentiable market models, weak $\mathcal{D}$-differentiable strategies, and, when necessary, with weak $\mathcal{D}$-differentiable state price deflators. All It\^{o} processes are weak $\mathcal{D}$-differentiable, so that the class of considered admissible strategies is very large.

\subsubsection{Arbitrage as Curvature}
The Lie algebra of $G$ is the function space of all real valued functions on $[0, +\infty [$ denoted by
\begin{equation*}\mathfrak{g}=\mathbb{R}^{[0, +\infty[}\end{equation*}
and therefore commutative.
Following Ilinski's idea proposed in \cite{Il01}, we motivate the choice of a particular $\mathfrak{g}$-valued connection $1$-form by the fact that it allows to encode portfolio rebalance (or foreign exchange) and discounting as parallel transport.
\begin{theorem}With the choice of connection
\begin{equation}\label{connection}\chi(x,t,g).(\delta x, \delta t):= \left(\frac{D_t^{\delta x}}{D_t^x}-r_t^x\delta t\right) g,\end{equation}
the stochastic parallel transport in $\mathcal{B}$ has the following financial interpretations:
\begin{itemize}
\item Parallel transport along the nominal directions ($x$-lines) corresponds to a multiplication by an exchange rate.
\item Parallel transport along the time direction ($t$-line) corresponds to a division by a stochastic discount factor.
\end{itemize}
\end{theorem}
\begin{proof}
We refer to Theorem 28 in \cite{Fa15}.
\end{proof}

Recall that time derivatives needed to define the parallel transport along the time lines have to be understood in Stratonovich's sense. We see that the bundle is trivial, because it has a global trivialization, but the connection is not trivial.
The connection $\chi$ writes as a linear combination of basis differential forms as
\begin{equation}
\chi(x,t,g)=\left(\frac{1}{D_t^x}\sum_{j=1}^ND_t^jdx_j-r_t^xdt\right)g.
\label{1-form}
\end{equation}
The $\mathfrak{g}$-valued curvature $2$-form is defined as
\begin{equation}R:=d\chi+[\chi,\chi],\end{equation} meaning by this, that for all $(x,t,g)\in \mathcal{B}$ and for all $\xi,\eta\in T_{(x,t)}M$
\begin{equation}R(x,t,g)(\xi,\eta):=d\chi(x,t,g)(\xi,\eta)+[\chi(x,t,g)(\xi),\chi(x,t,g)(\eta)]. \end{equation}
Remark that, being the Lie algebra commutative, the Lie bracket $[\cdot,\cdot]$ vanishes. After some calculations we obtain
\begin{equation}R(x,t,g)=\frac{g}{D_t^x}\sum_{j=1}^ND_t^j\left(r_t^x+\mathcal{D}\log(D_t^x)-r_t^j-\mathcal{D}\log(D_t^j)\right)dx_j\wedge dt,\end{equation}
\noindent summarized as the following.
\begin{proposition}[\textbf{Curvature Formula}]\label{curvature}
Let $R$ be the curvature. Then, the following equality holds:
\begin{equation}R(x,t,g)=g dt\wedge d_x\left[\mathcal{D} \log (D_t^x)+r_t^x\right].\end{equation}
\end{proposition}
The curvature represents the capacity of instantaneous arbitrage allowed by the market.
Although the original proof can be found in Proposition 38 in \cite{Fa15}, it is based on the physical concept such as the divergence and the current, which is not so familiar for mathematical finance, here we state afresh more straightforward proof.
\begin{proof}
Since the Lie bracket $[\cdot,\cdot]$ vanishes, and the exterior derivative $d$ acts only for the first term of the right hand side of (\ref{1-form}),
\begin{align*}
   R(x,t,g)  =  d\chi (x,t,g) = g \cdot d\Biggl(  \sum_{i=1}^N \frac{\partial \log (D_t^x)}{\partial x_i} \cdot dx_i - r_t^x dt    \Biggl) .
\end{align*}
We note that the differential $d$ acts as $d=d_x+d_t = d_x + \mathcal{D}$ for the first term, while the differential $d$ acts as $d=d_x$ for the second term $(-r^x_t)dt$ because $dt \wedge dt=0$ as bellow.
\begin{equation*}
d(-r^x_tdt) =d_x(-r^x_tdt) = -\sum_j \frac{\partial }{\partial x_j}r^x_t dx_j \wedge dt = \sum_j \frac{\partial }{\partial x_j} r^x_t dt \wedge dx_j.
\end{equation*}
And then we have
\begin{align*}
   R(x,t,g)   &= g \cdot \Biggl(  \sum_{i<j} \frac{\partial^2}{\partial x_i \partial x_j}
           \Bigl( \log (D_t^x) \Bigl) dx_i \wedge dx_j  \\
           &\quad\quad\quad\quad\quad +  \Biggl(    \sum_{j} \mathcal{D}\frac{\partial \log (D_t^x)}{\partial x_i}
                             + \sum_{j} \frac{\partial }{\partial x_j} r^x_t   \Biggl) dt \wedge dx_j \Biggl) ,
\end{align*}
but the first term vanish because of the anticommutativity of the wedge product $dx_i \wedge dx_j = - dx_j \wedge dx_i$. Rearrange the order of $\frac{\partial }{\partial x_i}$ and $\mathcal{D}$, we can conclude that
\begin{align*}
R(x,t,g)    &=  g \cdot   \sum_{j} \frac{\partial }{\partial x_j} \Bigl( \mathcal{D} \log (D_t^x)   +  r^x_t   \Bigl) dt \wedge dx_j   \\
   &= g \cdot dt \wedge d_x \Bigl(  \mathcal{D} \log (D_t^x)+  r^x_t    \Bigl) .
 \end{align*}
\end{proof}

\noindent We can prove following results which characterizes arbitrage as curvature.
\begin{theorem}[\textbf{No Arbitrage}]\label{holonomy}
The following assertions are equivalent:
\begin{itemize}
\item [(i)] The market model (consisting base assets and futures with discounted prices $D$ and $P$) satisfies the no-free-lunch-with-vanishing-risk condition.
\item[(ii)] There exists a positive local martingale $\beta=(\beta_t)_{t\ge0}$ such that deflators and short rates satisfy for all portfolio nominals and all times the condition
\begin{equation}r_t^x=-\mathcal{D}\log(\beta_tD_t^x).\end{equation}
\item[(iii)] There exists a positive local martingale $\beta=(\beta_t)_{t\ge0}$ such that deflators and term structures satisfy for all portfolio nominals and all times the condition
\begin{equation}P^x_{t,s}=\frac{\mathbb{E}_t[\beta_sD^x_s]}{\beta_tD^x_t}.\end{equation}
\end{itemize}
\end{theorem}
\begin{proof}
We refer to Theorem 33 in \cite{Fa15}.
\end{proof}
\noindent This motivates the following definition.
\begin{definition}
The market model satisfies the \textbf{zero curvature (ZC)} if and only if the curvature vanishes a.s.
\end{definition}

\noindent Therefore, we have following implication relying two different definitions of no-abitrage:
\begin{corollary}
\begin{equation*}
\text{(NFLVR)}\Rightarrow \text{(ZC)}.
\end{equation*}
\end{corollary}

\noindent As an example to demonstrate how the most important geometric concepts of Section \ref{section2} can be applied we consider an asset model whose dynamics is given by a multidimensional It\^{o} process. Let us consider a market consisting of  $N+1$ assets labeled by $j=0,1,\dots,N$, where the $0$-th asset is the cash account utilized as a num\'{e}raire. Therefore, as explained in the introductory Subsection \ref{StochasticPrelude}, it suffices to model the price dynamics of the other assets $j=1,\dots,N$ expressed in terms of the $0$-th asset. As vector valued semimartingales for the discounted price process $\hat{S}:[0,+\infty[\times\Omega\rightarrow\mathbb{R}^N$ and the short rate $r:[0,+\infty[\times\Omega\rightarrow\mathbb{R}^N$, we chose the multidimensional It\^{o} processes given by
\begin{equation}\label{Dyn}
\begin{split}
d\hat{S}_t&=\hat{S}_t(\alpha_tdt+\sigma_tdW_t)\\
dr_t&=a_tdt+b_tdW_t,
\end{split}
\end{equation}
where
\begin{itemize}
\item $(W_t)_{t\in[0,+\infty[}$ is a standard $\mathbb{P}$-Brownian motion in $\mathbb{R}^K$, for some $K\in\mathbb{N}$,
\item $(\sigma_t)_{t\in[0,+\infty[}$,  $(\alpha_t)_{t\in[0,+\infty[}$ are  $\mathbb{R}^{N\times K}$-, and respectively,  $\mathbb{R}^{N}$- valued  stochastic processes, $\sigma_t$ has maximal rank, i.e. $\text{rank}(\sigma_t)=K$, and,
\item $(b_t)_{t\in[0,+\infty[}$,  $(a_t)_{t\in[0,+\infty[}$ are  $\mathbb{R}^{N\times K}$-, and respectively,  $\mathbb{R}^{N}$- valued  stochastic processes.
\end{itemize}

\begin{proposition}\label{PropIto}
Let the dynamics of a market model be specified by following It\^o processes as in (\ref{Dyn}), where we additionally assume that  the coefficients
\begin{itemize}
\item $(\alpha_t)_t,(\sigma_t)_t$, and  $(r_t)_t$ satisfy
\begin{equation*}
\lim_{s\rightarrow t^-}\mathbb{E}_s[\alpha_t]=\alpha_t,\quad\lim_{s\rightarrow t^-}\mathbb{E}_s[r_t]=r_t,\quad\lim_{s\rightarrow t^-}\mathbb{E}_s[\sigma_t]=\sigma_t,
\end{equation*}
\item $(\sigma_t)_t$ is an It\^o process,
\item $(\sigma_t)_t$ and $(W_t)_t$ are independent processes.
\end{itemize}
Then, the market model satisfies the (ZC) condition if and only if
\begin{equation}\label{ZCCond}
\alpha_t+r_t \in {\rm Range} (\sigma_t).
\end{equation}
\end{proposition}
\begin{remark}In the case of the classical model, where there are no term structures (i.e. $r\equiv0$), the condition (\ref{ZCCond}) reads as $\alpha_t\in{\rm Range} (\sigma_t)$.
\end{remark}

\begin{proof}
Let us consider the expression for It\^o's integral with respect to Stratonovich's
\begin{equation*}
\int_0^t\sigma_udW_u=\int_0^t\sigma_u\circ dW_u-\frac{1}{2}\int_0^td\left<\sigma,W\right>_u,
\end{equation*}
and take Nelson's derivative corresponding to the Stratonovich's integral:
\begin{equation}
\mathcal{D}\int_0^t\sigma_udW_u=\sigma_t\mathcal{D}W_t-\frac{1}{2}\mathcal{D}\left<\sigma,W\right>_t.
\end{equation}
Since
\begin{equation}
\mathcal{D}W_t=\frac{W_t}{2t}
\end{equation}
and, because of the independence assumption for the two It\^{o} processes $(\sigma_t)_t$ and $(W_t)_t$,
\begin{equation*}
\left<\sigma,W\right>_t\equiv0,
\end{equation*}
we obtain
\begin{equation*}
\mathcal{D}\int_0^t\sigma_udW_u=\sigma_t\frac{W_t}{2t},
\end{equation*}
which, inserted into the asset dynamics
\begin{equation*}
\hat{S}_t=\hat{S}_0\exp\left(\int_0^t(\alpha_u-\frac{1}{2} {\rm diag} ({\sigma_u}\sigma_u^{\dagger}))du+\int_0^t\sigma_udW_u\right),
\end{equation*}
\noindent leads to
\begin{equation*}
\mathcal{D}\log\hat{S}_t=\alpha_t-\frac{1}{2} {\rm diag}({\sigma_t}\sigma_t^{\dagger})+\sigma_t\frac{W_t}{2t}.
\end{equation*}
By Proposition \ref{curvature} the curvature vanishes if and only if for all $x\in\mathbb{R}^N$
\begin{equation*}
\mathcal{D}\log \hat{S}_t^x+r_t^x=C_t,
\end{equation*}
for a real valued stochastic process $(C_t)_{t\ge0}$, or, equivalently
\begin{equation*}
\mathcal{D}\log \hat{S}_t+r_t=C_te,
\end{equation*}
where $e:=[1,\dots,1]^{\dagger}$ or
\begin{equation}\label{noarbalphabeta}
\alpha_t+r_t-\frac{1}{2} {\rm diag} (\sigma_t{\sigma_t}^{\dagger})+\sigma_t\frac{W_t}{2t}=C_te.
\end{equation}
Equation (\ref{noarbalphabeta}) is the formulation of the (ZC) condition for the market model (\ref{Dyn}).
By taking on both sides of (\ref{noarbalphabeta}) $\lim_{h\rightarrow0^+}\mathbb{E}_{t-h}[\cdot]$, and utilizing the independence assumption, from which
\begin{equation*}
\mathbb{E}_{t-h}\left[\sigma_t\frac{W_t}{2t}\right]=\mathbb{E}_{t-h}\left[\sigma_t\right]\underbrace{\mathbb{E}_{t-h}\left[\frac{W_t}{2t}\right]}_{=0}=0
\end{equation*}
follows, we obtain, using the continuity assumption for $(\alpha_t)_t,(\sigma_t)_t$, and  $(r_t)_t$,
\begin{equation*}
\alpha_t+r_t-\frac{1}{2} {\rm diag}(\sigma_t{\sigma_t}^{\dagger})=\beta_te,
\end{equation*}
where $\beta_t:=\lim_{h\rightarrow0^+}\mathbb{E}_{t-h}[C_t]$ is a predictable process. Therefore, equation (\ref{noarbalphabeta}) becomes
\begin{equation}
\sigma_t\frac{W_t}{2t}=(C_t-\beta_t)e,
\end{equation}
and, thus
\begin{equation}
e\in{\rm Range} (\sigma_t),
\end{equation}
the space spanned by the column vectors of $\sigma_t$. Since $\sigma_t$ has maximal rank, the $K$column vectors of $\sigma_t$ are linearly independent and $C_t-\beta_t\neq0$.\par Let $P_{\sigma_t}$, $P_{\sigma_t^\bot}$ denote the orthogonal projections onto ${\rm Range} (\sigma_t)$ and its orthogonal complement ${\rm Range} (\sigma_t)^{\bot}$, respectively. Then, we can decompose
\begin{equation}
\alpha_t+r_t=P_{\sigma_t}(\alpha_t+r_t)+P_{\sigma_t^\bot}(\alpha_t+r_t),
\end{equation}
and
\begin{equation}\label{eqP}
P_{\sigma_t^\bot}(\alpha_t+r_t)=P_{\sigma_t^\bot}\left(C_te\right)-P_{\sigma_t^\bot}\left(\sigma_t\frac{W_t}{2t}\right)+P_{\sigma_t^\bot}\left(\frac{1}{2} {\rm diag} (\sigma_t{\sigma_t}^{\dagger})\right).
\end{equation}
Since $e$ and $\sigma_tW_t$ lie in ${\rm Range} (\sigma_t)$, the first two addenda on the right hand side of (\ref{eqP}) vanish. By Lemmata \ref{diag} and \ref{proj} the third one vanishes as well, so that $P_{\sigma_t^\bot}(\alpha_t+r_t)=0$, i.e. $\alpha_t+r_t\in{\rm Range} (\sigma_t)$. Conversely, if  $\alpha_t+r_t\in{\rm Range} (\sigma_t)$, then equation (\ref{noarbalphabeta}) holds true, and
the proof of the equivalence between the (ZC) condition and (\ref{ZCCond}) is completed.
\end{proof}

\begin{lemma}\label{diag}
Let $A$ be a linear operator on the euclidean $\mathbb{R}^N$. The vector
\begin{equation*}
{\rm diag} (A):=\sum_{j=1}^N(Ae_j\cdot e_j)e_j
\end{equation*}
does not depend on the choice of the orthonormal basis $\{e_1,\dots,e_n\}$ of $\mathbb{R}^N$ and defines the \textbf{diagonal} of $A$.
\end{lemma}
\begin{proof}
The coordinates of ${\rm diag} (A)$ with respect to the orthonormal basis $\{e_1,\dots,e_N\}$ can be written as
\begin{equation}
[{\rm diag}(A)]_{\{e\}}=\sum_{j=1}^N([e_j]_{\{e\}}^\dagger[A]_{\{e\}}[e_j]_{\{e\}})[e_j]_{\{e\}}
\end{equation}
Let us consider another orthonormal basis $\{f_1,\dots,f_n\}$ of $\mathbb{R}^N$. This means that there exists an orthogonal linear operator $U$ on $\mathbb{R}^N$ such that $Ue_j=f_j$ for all $j=1,\dots,N$. Therefore we can write
\begin{equation}
\begin{split}
[{\rm diag}(A)]_{\{e\}}&=\sum_{j=1}^N\left(([U]_{\{f\}}^\dagger[f_j]_{\{f\}})^\dagger[A]_{\{e\}}[U]_{\{f\}}^\dagger[f_j]_{\{f\}}\right)[U]_{\{f\}}^\dagger[f_j]_{\{f\}}\\
&=\sum_{j=1}^N\left([f_j]_{\{f\}}^\dagger\left([U]_{\{f\}}[A]_{\{e\}}[U]_{\{f\}}^\dagger\right)[f_j]_{\{f\}}\right)[U]_{\{f\}}^\dagger[f_j]_{\{f\}}\\
&=[U]_{\{f\}}^\dagger\left(\sum_{j=1}^N[f_j]_{\{f\}}^\dagger[A]_{\{f\}}[f_j]_{\{f\}}\right)\\
&=[U]_{\{f\}}^\dagger[{\rm diag}(A)]_{\{f\}}.
\end{split}
\end{equation}
Therefore, the coordinates of the diagonal transforms like a vector during a change of basis, and hence the diagonal is well defined.\\
\end{proof}

\begin{lemma}\label{proj}
Let $\sigma$ be a $\mathbb{R}^{N\times K}$ real matrix of rank $K$ and $P$ the orthogonal projection onto the orthogonal complement to the subspace generated by the column vectors of $\sigma$. Then,
\begin{equation*}
P{\rm diag}(\sigma\sigma^{\dagger})=0\in\mathbb{R}^N.
\end{equation*}
\end{lemma}
\begin{proof}
The real symmetric matrix $\sigma\sigma^{\dagger}\in\mathbb{R}^{N\times N}$ induces via standard orthonormal basis a selfadjoint linear operator on $\mathbb{R}^N$, which by Lemma \ref{diag} has a well defined diagonal. Let us enlarge $\sigma$ to an $\mathbb{R}^{N\times N}$ matrix, by adding $N-K$ zero column vectors. The matrix $\sigma\sigma^{\dagger}\in\mathbb{R}^{N\times N}$ remains the same. Let us consider an orthonormal basis of $\mathbb{R}^N$, $\{f_1,\dots,f_N\}$, where $\{f_1,\dots,f_K\}$ is a basis of $\text{Range}(\sigma)$ and $\{f_{K+1},\dots,f_N\}$ is a basis of its orthogonal complement, $\text{Range}(\sigma)^\bot$. The diagonal of $\sigma\sigma^{\dagger}$ reads
\begin{equation}
{\rm diag}(\sigma\sigma^{\dagger})=\sum_{j=1}^N(\sigma\sigma^{\dagger}f_j\cdot f_j)f_j=\sum_{j=1}^N(\sigma^{\dagger}f_j\cdot \sigma^{\dagger}f_j)f_j=\sum_{j=1}^K(\sigma^{\dagger}f_j\cdot \sigma^{\dagger}f_j)f_j,
\end{equation}
because $\sigma^{\dagger}f_j=0$ for $j=K+1,\dots,N$, being $f_j$ in the orthogonal complement of $\text{Range}(\sigma)$. Therefore,
\begin{equation}
P{\rm diag}(\sigma\sigma^{\dagger})=\sum_{j=1}^K(\sigma^{\dagger}f_j\cdot \sigma^{\dagger}f_j)Pf_j=0,
\end{equation}
because $f_j$ is in $\text{Range}(\sigma)$ for $j=1,\dots,K$ and $P$ is the projection onto $\text{Range}(\sigma)^\bot$ .
\end{proof}

\noindent Next, we show the equivalence of the (ZC) condition with (NFLVR) in the case of It\^o's dynamics.
\begin{proposition}\label{NovikovThm}
Under the same assumptions as Proposition \ref{PropIto}, the zero curvature condition for the market model specified by (\ref{Dyn}) , that is
\begin{equation*}
\mathcal{D}\log \hat{S}_t+r_t =C_te,
\end{equation*}
is equivalent to the no-free-lunch-with-vanishing-risk condition if the positive stochastic process $(\beta_t)_{t\ge0}$, defined as
\begin{equation*}
\beta_t:=\exp\left(-\int_0^tC_udu\right)
\end{equation*}
is a local martingale.
\end{proposition}
\begin{proof}
By Proposition \ref{curvature} the zero curvature (ZC) condition $R=0$ is equivalent with the existence of a stochastic process $(C_t)_{t\ge0}$ such that for all $i=1,\dots,N$ the equation
\begin{equation*}
\mathcal{D}\log \hat{S}_t^i+r_t ^i=C_t
\end{equation*}
\noindent holds. This means that
\begin{equation*}
\begin{split}
&\mathcal{D}\log \hat{S}_t^i=C_t-r_t^i\\
&\log\frac{S_t^i}{S_0^i}=\int_0^t(C_u-r_u^i)du\\
&S_t^i=S_0^i\exp\left(\int_0^tC_udu\right)\exp\left(-\int_0^tr_u^idu\right).
\end{split}
\end{equation*}
Therefore,
\begin{equation*}
\mathcal{D}\log(\beta_tD_t^i)+r_t^i=0
\end{equation*}
for all $i=1,\dots,N$. By Theorem \ref{holonomy}, if $(\beta_t)_{t\ge0}$ is a martingale, then we have proved (NFLVR).
\end{proof}

\noindent We can reformulate the result of Proposition \ref{PropIto} as follows.
\begin{corollary}\label{CorRho}
Let $\{J_t^1,\dots,J_t^{B}\}$ be an orthonormal basis of $\ker(\sigma_t)\subset\mathbb{R}^N$. Under the same assumptions as Proposition \ref{PropIto} the (ZC) condition for the market model (\ref{Dyn}), which is equivalent to (NFLVR), is equivalent to
\begin{equation}\label{rho}
\rho_t:=J_t^{\dagger}(\alpha_t+r_t)\equiv0\in\mathbb{R}^{B},
\end{equation}
where $J_t:=[J_t^1,\dots,J_t^B]$.
\end{corollary}

\begin{remark}[\textbf{Counterexamples}]\label{remContA}
Let us consider a financial market with a cash account with $r_t^0=0$ and a single risky asset  with (discounted) price  given by
\begin{equation}\label{contA}
S_t = e^{X_t},\text{ where } X_t:=\int_0^t\frac{W_u}{u}+W_t,
\end{equation}
for a standard univariate Brownian motion $(W_t)_t$. By It\^{o}'s formula, it folllows that
\begin{equation*}
dS_t = S_t\left(\frac{1}{2}+\frac{W_t}{t}\right)+S_tdW_t\text{ and }S_0=1.
\end{equation*}
In the notation of Proposition \ref{PropIto}, this corresponds to $\alpha_t = \frac{1}{2}+\frac{W_t}{t}$ and $r_t=0$.
The coefficient $\alpha_t$ does not satisfy the assumption $\lim_{s\rightarrow t^-}\mathbb{E}_s[\alpha_t]=\alpha_t$, because
\begin{equation*}
\lim_{s\rightarrow t^-}\mathbb{E}_s\left[\frac{1}{2}+\frac{W_t}{t}\right]=\frac{1}{2}+\frac{1}{2}\lim_{s\rightarrow t^-}\underbrace{\mathbb{E}_s[W_t]}_{=0}\neq\alpha_t.
\end{equation*}
The process $(S_t)_t$ does not satisfy (NFLVR), since $\int_0^{\epsilon}\left(\frac{W_t}{t}\right)^2dt>0$ for all $\epsilon>0$ as a consequence of Corollary 3.2 of \cite{JeYo79}. In the terminology of Delbaen \& Schachermayer, model (\ref{contA}) generates immediate arbitrage opportunities. Other
simple counterexamples can be constructed from Brownian bridges, which provide well known
examples of models admitting arbitrage (see \cite{LoWi00}).\par
Moreover, Fontana \& Runggaldier present asset models in \cite{Fo15} (Example 7.5) and \cite{FoRu13} (page 59) based on Bessel processes, which do not fulfill the assumptions of Propositions \ref{PropIto} and \ref{NovikovThm}. They are an example of dynamics satisfying (NUPBR) but not (NFLVR). The proof of the (NUPBR) property is based on its equivalence with the non-existence of arbitrage possibilities of the first kind.
\end{remark}

%
%
%
%
\section{Arbitrage and Utility}
Let us now consider a utility function, that is a real $C^2$-function of a real variable, which is strictly monotone increasing (i.e. $u^{\prime}>0$) and concave (i.e. $u^{\prime\prime}<0$). Typically, a market participant would like to maximize the expected utility of its wealth at some time horizon. Let us assume that he (or she) holds a portfolio of synthetic zero bonds delivering at maturity base assets and that the time horizon is infinitesimally near, that is that the utility of the instantaneous total return has to be maximized. The portfolio values read as:
\begin{itemize}
\item At time $t-h$: $D^x_{t-h}P^x_{t-h,t+h}$.
\item At time $t$: $D^x_tP^x_{t,t+h}$.
\item At time $t+h$:  $D^x_{t+h}$.
\end{itemize}
\noindent From now on we make the following\\
\textbf{Assumptions:}
\begin{enumerate}
\item[\textbf{(A1):}] The market filtration $(\mathcal{A}_t)_{t\ge0}$ is the coarsest filtration for which $(D_t)_{t\ge0}$ is adapted.
\item[\textbf{(A2):}] The process $(D_t)_{t\ge0}$ is Markov with respect to the filtration $(\mathcal{A}_t)_{t\ge0}$.
\end{enumerate}

\begin{proposition}
Under the assumptions $(A1)$ and $(A2)$ the synthetic bond portfolio instantaneous return can be computed as:
\begin{equation*}
\text{Ret}_t^ x:=\lim_{h \rightarrow 0^+}\mathbb{E}_{t}\left[\frac{D^x_{t+h}-D^x_{t-h}P^x_{t-h,t+h}}{2hD^x_{t-h}P^x_{t-h,t+h}}\right]=\mathcal{D}\log (D^x_t)+r^x_t.
\end{equation*}
\end{proposition}
\begin{proof} Under the assumptions $(A1)$ and $(A2)$ the conditional expectations with respect to the market filtration $(\mathcal{A}_t)_{t\ge 0}$  are the same as those computed with respect to the present $(\mathcal{N}_t)_{t\ge 0}$, past $(\mathcal{P}_t)_{t\ge 0}$ and future $(\mathcal{F}_t)_{t\ge 0}$ filtrations (see Appendix \ref{Derivatives}). Therefore, we can develop the instantaneous return as
\begin{equation*}
\begin{split}
&\lim_{h \rightarrow0^+}\mathbb{E}_{t}\left[\frac{D^x_{t+h}-D^x_{t-h}P^x_{t-h,t+h}}{2hD^x_{t-h}P^x_{t-h,t+h}}\right]\\
&=\lim_{h \rightarrow0^+}\mathbb{E}_{t}\left[\frac{D^x_{t+h}-D^x_{t-h}}{2hD^x_{t-h}P^x_{t-h,t+h}}+\frac{1-P^x_{t-h,t+h}}{2hP^x_{t-h,t+h}}\right]\\
&=\frac{1}{D_t^x}\mathcal{D}D_t^x+\lim_{h \rightarrow0^+}\frac{\exp\left(\int_{t-h}^{t+h}ds\,f_{t-h,s}^x\right)-1}{2h}=\mathcal{D}\log D_t^x + r_t^x.
\end{split}
\end{equation*}
\end{proof}
\begin{remark} This portfolio of synthetic zero bonds in the theory corresponds to a portfolio of futures in practice. If the short rate  vanishes, then the future corresponds to the original asset.
\end{remark}
\begin{definition}[\textbf{Expected Utility of Synthetic Bond Portfolio Return}]\label{expU} Let $t\ge s$ be fixed times. The expected utility maximization problem at time $s$ for the horizon $T$ for initial capital $\xi$ writes
\begin{equation}\label{opt}
        \sup_{\substack{x=\{x_h\}_{h\ge s}\\D_s^{x_s}=\xi}}\mathbb{E}_s\left[u\left(\exp\left(\int_s^Tdt\,\left(\mathcal{D}\log (D^{x_t}_t)+r^{x_t}_t\right)\right)D_s^{x_s}P_{s,T}^{x_s}\right)\right],
\end{equation}
where the supremum is taken over all weak $\mathcal{D}$-differentiable self-financing admissible strategies $x=\{x_u\}_{u\ge0}$.
\end{definition}
\noindent Now we can formulate the first result of this subsection.

\begin{theorem}\label{utility}
Let us assume $(A1)$ and $(A2)$, and that $(D_t)_t, (r_t)_t$ are weakly $\mathcal{D}$-differentiable semimartingales. The market curvature vanishes if and only if the expected utility maximization problem can be solved for all times and horizons for a chosen utility function.
\end{theorem}
\noindent This result can be seen as the natural generalization of the corresponding result in discrete time, as Theorem 3.5 in \cite{FoeSc04}, see also \cite{Ro94}. Compare with Bellini's, Frittelli's and Schachermayer's results for infinite dimensional optimization problems in continuous time, see Theorem 22 in \cite{BeFr02} and Theorem 2.2 in \cite{Scha01}. Nothing is said about the fulfilment of  the no-free-lunch-with-vanishing-risk condition: only the weaker zero curvature condition is equivalent to the maximization of the expected utility at all times for all horizons.
\begin{proof} The optimization problem (\ref{opt}) into a standard problem of stochastic optimal theory in continuous time which can be solved by means of a fundamental solution of the Hamilton-Jacobi-Bellman partial differential equation.\par However, there is a direct method, using Lagrange multipliers for Banach spaces (see \cite{Lu69} pages 239–270 and \cite{Ze95} Section 4.14, pages 270–271). First, remark that problem (\ref{expU}) is a concave optimization problem with convex domain and concave utility function and has therefore a unique solution corresponding to a global maximum. The Lagrange principal function corresponding to the maximum problem
\begin{equation}\label{end}
\begin{split}
\Phi(x,\lambda,\mu)&:=\mathbb{E}_s\left[u\left(\exp\left(\int_s^Tdt\,\left(\mathcal{D}\log (D^{x_t}_t)+r^{x_t}_t\right)\right)D_s^{x_s}P_{s,T}^{x_s}\right)\right. \\
&\left.\qquad\qquad-\int_s^Tdt\,\lambda_t\mathfrak{D}x_t\cdot
D_t\right]-\mu(D_s^{x_s}-\xi).
\end{split}
\end{equation}
Note that the Lagrange multiplier $\lambda$ corresponding to the self-financing condition (\ref{54}), expressed in terms of Nelson's derivative corresponding to It\^{o}'s differential, is a stochastic process $(\lambda_t)_{t\ge0}$. This Lagrange multiplier is weak $\mathcal{D}$-differentiable as all process involved so far are. The Lagrange multiplier $\mu$ corresponding to the initial wealth is a real number. To solve the maximization problem for $\Phi$ with respect to the processes $(x_t)$ and $(\lambda_t)$ we embed the optimal solution into a one parameter family as
\begin{equation*}
\begin{cases}
x_t(\epsilon)&:=x_t+\epsilon \delta x_t\\
\lambda_t(\eta)&:=\lambda_t+\eta \delta \lambda_t\\
\mu(\nu)&:=\mu+\nu\delta \mu,
\end{cases}
\end{equation*}
where $\epsilon$, $\eta$ and $\nu$ are real parameters defined in a neighborhood of $0$, and $\delta x_t$, $\delta \lambda_t$ and $\delta \nu$ are arbitrary variations such that the boundary conditions
\begin{equation}\label{bc}
\begin{cases}
&x_s(\epsilon)\equiv x_s\\
&x_T(\epsilon)\equiv x_T ,
\end{cases}
\end{equation}
are satisfied. The Lagrange principal equations associated to this maximization problem read

\begin{align}\label{LPE}
\begin{cases}
\displaystyle \frac{\partial \Phi}{\partial \epsilon} \Biggl|_{\epsilon=\eta=\nu:=0} &\displaystyle = \ \mathbb{E}_s\Biggl[ u^{\prime}\left(\exp\left(\int_s^Tdt\,\left(\mathcal{D}\log (D^{x_t}_t)+r^{x_t}_t\right)\right)D_s^{x_s}P_{s,T}^{x_s}\right) \\
&\qquad\qquad \displaystyle \cdot\exp\left(\int_s^Tdt\,\left(\mathcal{D}\log (D^{x_t}_t)+r^{x_t}_t\right)\right)D_s^{x_s}P_{s,T}^{x_s} \\
&\qquad\qquad \displaystyle \cdot\int_s^Tdt\frac{\partial}{\partial x}\left(\mathcal{D}\log (D_t^{x})+r_t^{x}\right)\Bigl|_{x=x_t}\cdot\delta x_t\\
&\qquad\qquad\qquad\qquad\qquad\qquad \displaystyle -\int_s^Tdt\,\lambda_t\mathfrak{D}\delta x_t\cdot D_t \Biggl]-\mu D_s^{\delta x_s} \ = \ 0 \\
 \displaystyle  \left.\frac{\partial \Phi}{\partial \eta}\right|_{\epsilon=\eta=\nu:=0} &\displaystyle = \ -\int_s^Tdt\,\delta\lambda_t\mathfrak{D}x_t\cdot D_t \ = \ 0\\
\displaystyle  \left.\frac{\partial \Phi}{\partial \nu}\right|_{\epsilon=\eta=\nu:=0} &\displaystyle = -\delta\mu(D_s^{x_s}-\xi) \ = \ 0,
\end{cases}
\end{align}
\noindent where, by Leibniz's theorem, we have interchanged differentiation with respect to $\epsilon$ or $\eta$ with the integration with respect to $t$ and the conditional expectation. The boundary condition implies $\delta x_s=0$, and hence $\mu D_s^{\delta x_s}=0$.\\ Integration by parts with respect to the time variable shows that
\begin{equation*}
-\int_s^Tdt\,\lambda_t\mathfrak{D}\delta x_t\cdot D_t=\int_s^Tdt\,\mathfrak{D(}\lambda_tD_t)\cdot\delta x_t,
\end{equation*}
\noindent which, inserted into the first equation of (\ref{LPE}) leads to
\begin{equation}\label{eqE}
\mathbb{E}_s\left[\int_s^Tdt\left(\left.M\frac{\partial}{\partial x}\left(\mathcal{D}\log (D_t^{x})+r_t^{x}\right)\right|_{x=x_t}+\mathfrak{D}(\lambda_tD_t)\right)\cdot\delta x_t\right]=0,
\end{equation}
\noindent where
\begin{equation*}
\begin{split}
M& :=u^{\prime}\left(\exp\left(\int_s^Tdt\,\left(\mathcal{D}\log (D^{x_t}_t)+r^{x_t}_t\right)\right)D_s^{x_s}P_{s,T}^{x_s}\right) \\
& \quad\quad\quad \cdot \exp\left(\int_s^Tdt\,\left(\mathcal{D}\log (D^{x_t}_t)+r^{x_t}_t\right)\right)D_s^{x_s}P_{s,T}^{x_s}
\end{split}
\end{equation*}
is a strictly positive random variable. Since the variation $\delta x_t$ is arbitrary we infer from (\ref{eqE})
\begin{equation*}
\left.M\frac{\partial}{\partial x}\left(\mathcal{D}\log (D_t^{x})+r_t^{x}\right)\right|_{x=x_t}+\mathcal{D}(\lambda_tD_t)=0\quad\text{ for any } t\in[s,T],
\end{equation*}
\noindent and, thus, for the choice $t:=s$,  it
follows, being the initial condition $x_s\in\mathbb{R}^N$ arbitrary
\begin{equation*}
\mathcal{D}\log (D_t^{x})+r_t^{x}=-\frac{1}{M}\mathfrak{D}(\lambda_tD_t^j)x_j+C_t^j \text{ for all }j=1,\dots, N, \end{equation*}
\noindent for a stochastic process $(C_t^j)_{t\ge0}$. Therefore
\begin{equation*}
-\frac{1}{M}\mathfrak{D}(\lambda_tD_t^j)x_j+C_t^j=-\frac{1}{M}\mathfrak{D}(\lambda_tD_t^i)x_i+C_t^i\text{ for all }j\neq i,
\end{equation*}
\noindent which can hold true if and only if
\begin{equation}
\begin{cases}
\mathfrak{D}(\lambda_tD_t^j)=0\\
C_t^j=C_t
\end{cases}
\end{equation}
\noindent for all $j=1,\dots,N$. Hence, for the optimal Lagrange multiplier,
\begin{equation*}
\mathfrak{D}(\lambda_tD_t)=0\in\mathbb{R}^N,
\end{equation*}
\noindent and
\begin{equation}\label{init}\mathcal{D}\log (D_t^{x})+r_t^{x}=C_t \text{ for all }j=1,\dots, N. \end{equation}
\noindent Therefore,  by Proposition \ref{curvature}, the curvature must vanish, which means that the existence of a solution to the maximization problem implies the vanishing of the curvature. The converse is also true, as it can be seen by following back the steps in this proof from (\ref{init}) to (\ref{end}). Hence, the equivalence between (ZC) and (\ref{opt}) holds true.\\
\end{proof}

\noindent It turns out that the two weaker notions of arbitrage, the zero curvature and the no-unbounded-profit-with-bounded-risk are equivalent.
\begin{theorem}\label{thm_ZC_equiv}
Let us assume $(A1)$ and $(A2)$, and that $(D_t)_t, (r_t)_t$ are semimartingales. Then,
\begin{equation*}
\text{(NUPBR)}\Rightarrow{(ZC)}.
\end{equation*}
\end{theorem}
\begin{remark}[\textbf{Counterexample}]
The model given by (\ref{contA}) satisfies (ZC), because it is one dimensional in the risky assets, but does not fulfill (NUPBR), because it allows for immediate arbitrage opportunities as shown in Remark \ref{remContA}.
\end{remark}
\begin{proof}[\text{Proof of Theorem \ref{thm_ZC_equiv}}]
By Proposition 2.1 (4) in \cite{HuSc10} the (NUPBR) is equivalent  with the existence of a growth optimal portfolio. We apply the classic set up of portfolio optimization to the portfolio of futures under consideration, (which covers as a special case the portfolio of base assets). Since the value of the portfolio at time $s$ is
\begin{equation*}
D_s^{x_s}P_{s,T}^{x_s},
\end{equation*}
and the growth factor from $s$ to $T$ is
\begin{equation*}
\exp\left(\int_s^Tdt\,\left(\mathcal{D}\log (D^{x_t}_t)+r^{x_t}_t\right)\right),
\end{equation*}
the solution of the expected utility maximization for $s:=0$ and arbitrary $T$ with utility function $u:=\log$ must be equal to the optimal growth portfolio. Therefore, by Theorem \ref{utility} (ZC) follows.
\end{proof}

\noindent Under what conditions is the converse of Theorem \ref{thm_ZC_equiv} true?
The equivalence of expected utility maximization and (NFLVR) can be proved for a particular choice of a Markov dynamics. Namely, if the asset dynamics follows an It\^{o} process, Proposition \ref{NovikovThm} and Theorem \ref{utility} lead to
\begin{proposition}\label{convNUPBR}
Let the dynamics of a market model be specified by following It\^o processes as in (\ref{Dyn}), where we additionally assume that  the coefficients
\begin{itemize}
\item $(\alpha_t)_t,(\sigma_t)_t$, and  $(r_t)_t$ satisfy
\begin{equation*}
\lim_{s\rightarrow t^-}\mathbb{E}_s[\alpha_t]=\alpha_t,\quad\lim_{s\rightarrow t^-}\mathbb{E}_s[r_t]=r_t,\quad\lim_{s\rightarrow t^-}\mathbb{E}_s[\sigma_t]=\sigma_t,
\end{equation*}
\item $(\sigma_t)_t$ is an It\^o process,
\item $(\sigma_t)_t$ and $(W_t)_t$ are independent processes.
\end{itemize}
Then,  the (NFLVR) condition holds true if and only if the expected utility maximization problem can be solved for all times and horizons for a chosen utility function.
\end{proposition}

\begin{corollary}
Under the same assumptions of Proposition \ref{convNUPBR},
\begin{equation*}
\text{(ZC)}\Rightarrow \text{(NUPBR)}.
\end{equation*}
\end{corollary}
\noindent These last two results are in line with the well-known results of \cite{Be01}, \cite{ChLa07}, \cite{KaKa07} and \cite{HuSc10}.

%
%
%
%
\section{Arbitrage and Derivative Pricing}
The (NFLVR) ia an equilibrium condition for financial markets, and the Black-Scholes PDE allows for a unique pricing of derivatives of the base assets of those financial markets. Even if the (ZC) is not fulfilled, the market forces determine an asset dynamics minimizing the total quantity of arbitrage allowed by the market, as it was shown in \cite{Fa15, Fa21}. The minimal arbitrage is an equilibrium condition as well, which generalizes the benchmark approach (e.g. \cite{HePl06}) leading to a probability measure equivalent to the statistical one, which is the best possible approximation for the risk neutral measure (cf. the forthcoming \cite{FaTa20}). In this case too, a (non-linear) PDE allows for a unique pricing of derivatives of the base assets, in which the arbitrage measure explicitly appears.
\subsection{The Black-Scholes PDE in the Presence of Arbitrage}
For markets allowing for arbitrage we are in the position to derive the price dynamics of derivatives whose underlying following an It\^{o} process. It is a non linear partial differential equation which coincides with the linear Black-Scholes partial differential equation as soon as the arbitrage vanishes.
\begin{theorem}\label{ThmBS}
Let us consider a market consisting in a bank account, an asset and a derivative  whose discounted prices $X_t$ and $\Phi(t,X_t)$ follow an It\^{o}'s process. In particular
\begin{equation*}
dX_t=X_t(\alpha_tdt+\sigma_t dW_t),
\end{equation*}
where $(\alpha_t)_{t\in[0,+\infty[}$ and $(\sigma_t)_{t\in[0,+\infty[}$ are real valued adapted processes, the latter with finite variation.
Assuming that the pay-off function $\Phi=\Phi(t,x)\in C^{1,2}$, the derivative discounted price solves the PDE
\begin{equation}\label{NonLinearBS}
\frac{\partial\Phi}{\partial t}+\frac{\sigma_t^2}{2}X_t^2\frac{\partial^2\Phi}{\partial x^2}=\rho_t\Phi\left(1+\left(\frac{1}{\Phi}\frac{\partial\Phi}{\partial x}X_t\right)^2\right)^{\frac{1}{2}},
\end{equation}
where $\rho_t$, defined in (\ref{rho}) measures the arbitrage allowed by the market.
\end{theorem}
\begin{proof}
We prove this theorem in the context of Corollary \ref{CorRho} with vanishing short rate $r_t$. By assumption, choosing $N:=2$ and $B:=1$, the market dynamics reads
\begin{equation}\label{dynBS}
d\hat{S}_t=\hat{S}_t(\bar{\alpha}_tdt+\bar{\sigma}_tdW_t),
\end{equation}
where
\begin{equation*}
\hat{S}_t:=\left[
             \begin{array}{c}
               X_t \\
               \Phi(t,X_t)\\
             \end{array}
           \right],\quad
\bar{\alpha}_t:=\left[
             \begin{array}{c}
               \alpha_t \\
               \beta_t\\
             \end{array}
           \right],\quad
\bar{\sigma}_t:=\left[
             \begin{array}{c}
               \sigma_t \\
               \tau_t\\
             \end{array}
           \right].
\end{equation*}
for appropriate real valued predictable processes  $(\beta_t)_{t\in[0,+\infty[}$ and $(\tau_t)_{t\in[0,+\infty[}$ characterizing the dynamics of the derivative.
We apply It\^{o}'s Lemma to the second component of (\ref{dynBS}). By comparing deterministic and stochastic terms we obtain
\begin{equation}\label{BSS}
\left\{
  \begin{array}{ll}
    {\displaystyle \frac{\partial\Phi}{\partial t}+\frac{\partial\Phi}{\partial x}X_t\alpha_t+\frac{\sigma_t^2}{2}\frac{\partial^2\Phi}{\partial x^2}X_t^2=\beta_t\Phi }\\ \\
    {\displaystyle \frac{\partial\Phi}{\partial x}X_t\sigma_t=\tau_t\Phi.}
  \end{array}
\right.
\end{equation}
The one dimensional $\ker(\bar{\sigma}_t)$ is spanned by
\begin{equation}
J_t:=(\sigma_t^2+\tau_t^2)^{-\frac{1}{2}}\left[
             \begin{array}{c}
               -\tau_t \\
               +\sigma_t\\
             \end{array}
           \right],
\end{equation}
and the vector $\bar{\alpha}_t$ admits the decomposition
\begin{equation}\label{alphadecomp}
\bar{\alpha}_t=\lambda_t\bar{\sigma}_t+\rho_tJ_t,
\end{equation}
for reals $\lambda_t$ and $\rho_t=\bar{\alpha}_t^{\dagger}J_t$. Now we can insert (\ref{alphadecomp}) into (\ref{BSS}) and eliminate $\lambda_t$, since the $\lambda_t$ terms cancel out. The first equation of (\ref{BSS}) becomes
\begin{equation}\label{firstBS}
\frac{\partial\Phi}{\partial t}+\frac{\sigma_t^2}{2}X_t^2\frac{\partial^2\Phi}{\partial x^2}=\rho_tX_t\frac{\partial\Phi}{\partial x}\left(\frac{\sigma_t^2+\tau_t^2}{\tau_t^2}\right)^{\frac{1}{2}}.
\end{equation}
The second equation of (\ref{BSS}) can be written as
\begin{equation*}
\frac{\sigma_t}{\tau_t}=\frac{\Phi}{X_t\frac{\partial \Phi}{\partial x}},
\end{equation*}
\noindent which, inserted into (\ref{firstBS}) gives (\ref{NonLinearBS}).\\
\end{proof}
\begin{remark}
In \cite{FaVa12}, utilizing another measure of arbitrage $\tilde{\rho}_t$, the PDE
\begin{equation}\label{BSSFaVa}
\frac{\partial\Phi}{\partial t}+\frac{\sigma_t^2}{2}X_t^2 \frac{\partial^2\Phi}{\partial x^2}
=-\sqrt{2}\tilde{\rho}_t \Phi\left[1+\frac{X_t^2}{\Phi^2}\left(\frac{\partial \Phi}{\partial x}\right)^2-\frac{\partial \Phi}{\partial x}\right]^{\frac{1}{2}},
\end{equation}
\noindent was derived. After some computations, it turns out that
\begin{equation*}
\tilde{\rho}_t=-\frac{1}{\sqrt{2}}\left(\frac{1+\frac{X_t^2}{\Phi^2}\left(\frac{\partial \Phi}{\partial x}\right)^2}{1-\frac{X_t}{\Phi}\frac{\partial \Phi}{\partial x}+\frac{X_t^2}{\Phi^2}\left(\frac{\partial \Phi}{\partial x}\right)^2}\right)^{\frac{1}{2}}\rho_t,
\end{equation*}
thus guaranteeing that both (\ref{BSS}) and (\ref{BSSFaVa}) are two representations of the same non linear Black-Scholes PDE for the price of a derivative in the presence of arbitrage.
\end{remark}
\begin{remark} If arbitrage possibilities are allowed, there is no risk neutral probability measure. Asset pricing can nevertheless be obtained as (conditional) expectation of discounted asset's cash flows with respect to the \textit{minimal arbitrage probability measure}, as explained in \cite{FaTa20}.
\end{remark}

\noindent It is possible to reformulate Theorem \ref{ThmBS} directly in terms of prices and not discounted prices.
\begin{corollary}\label{CorNonLinearBSCF}
Let us consider a market consisting in a bank account with constant instantaneous risk free rate $r$, an asset and a derivative  whose prices $S_t$ and $\Psi(t,S_t)$ follow an It\^{o} process. In particular
\begin{equation*}
dS_t=S_t(\alpha_tdt+\sigma_t dW_t),
\end{equation*}
where $(\alpha_t)_{t\in[0,+\infty[}$ and $(\sigma_t)_{t\in[0,+\infty[}$ are real valued adapted processes, the latter with finite variation.
Assuming that the pay-off function $\Psi=\Psi(t,s)\in C^{1,2}$, the derivative price solves the PDE
\begin{equation}\label{NonLinearBSCF}
\frac{\partial\Psi}{\partial t}+rS_t\frac{\partial \Psi}{\partial s}+\frac{\sigma_t^2}{2}S_t^2\frac{\partial^2\Psi}{\partial s^2}-r\Psi=\rho_t\Psi\left(1+\left(\frac{1}{\Psi}\frac{\partial\Psi}{\partial s}S_t\right)^2\right)^{\frac{1}{2}},
\end{equation}
where $\rho_t$, defined in (\ref{rho}) measures the arbitrage allowed by the market.
\end{corollary}
\noindent Note that in the (ZC) case (\ref{NonLinearBSCF}) becomes the celebrated linear Black-Scholes PDE well known from textbooks.
\begin{proof}
In the equation (\ref{NonLinearBS}) we insert
\begin{equation*}
\left\{
  \begin{array}{l}
    \Phi(t,x)=\Psi(t,s)e^{-rt}\\
    x=e^{-rt}s,
  \end{array}
\right.
\end{equation*}
\noindent and, taking into account that
\begin{equation*}
\frac{\partial}{\partial x}=e^{rt}\frac{\partial}{\partial s}\qquad\qquad
\frac{\partial^2}{\partial x^2}=e^{2rt}\frac{\partial^2}{\partial s^2}\qquad\qquad
\frac{\partial s}{\partial t}=rs,
\end{equation*}
we obtain, after some algebra equation (\ref{NonLinearBSCF}).
\end{proof}

\subsection{Approximate Solution of the Modified Black-Scholes PDE}
In this subsection we derive a dependence relation between a call option price, the price of its underlying  and the arbitrage measure $\rho$ in an implicit form. For this purpose, we assume that the arbitrage measure $\rho_t \equiv \rho$ is constant during the period considered, typically between $0$ and the derivative maturity $T$.
As \cite{FaVa12} discussed empirically, arbitrage measure is relatively small so we consider perturbations with respect to $\rho$ and seek an approximate solution of the modified Black-Scholes PDE (\ref{NonLinearBS}). We note that the non linear term of the modified Black-Scholes PDE (\ref{NonLinearBS}) is multiplied by $\rho$ linearly.
\begin{theorem}
For sufficiently small $\rho>0$, an approximated solution of the modified Black-Scholes PDE (\ref{NonLinearBS}) under the terminal condition  $\Phi(T,X_T) = (X_T- K)^+$, where $K$ is the strike price at time $T$ on the discounted value of the underlying with constant volatility $\sigma$, is given by
\begin{equation}\label{solcall}
\Phi(t,X_t) = Ke^{\frac{1}{2}\log \frac{X_t}{K}-\frac{1}{8}\sigma^2 (T-t)} u\left(\frac{1}{2}\sigma^2 (T-t),\log \frac{X_t}{K}\right),
\end{equation}
\noindent where
\begin{equation*}
u(\tau,y)=u_0(\tau,y)+\rho U_1(\tau,y)+\rho^2 U_2(\tau,y)+O(\rho^3)\quad(\rho\rightarrow0)
\end{equation*}
and $u_0(\tau,y)$ is the solution of $(\partial _{\tau } - \partial^2_y )u_0(\tau , y) =  0$ with the initial condition $u(0,y) = \max \{ e^{\frac{y}{2}}-e^{-\frac{y}{2}} ,0 \} $, and
\begin{equation}\label{defu1}
\begin{split}
f(v_1,v_2)&:=\frac{2K}{\sigma^2}\sqrt{\frac{5}{4}v_1^2 +v_1 v_2+v_2^2}\\
G(\tau , y;s,z) &:= \frac{1}{2\sqrt{\pi (\tau -s) }} \exp \left( -\frac{(y-z)^2 }{4 (\tau -s)}\right)\\
U_1(\tau ,y)&:= \int_0^{\tau }ds\int_{-\infty }^{\infty }dz\,G(\tau , y;s,z)f(u_0(s,z),u^{\prime}_0(s,z))\\
U_2(\tau ,y)&:= \int_0^{\tau } ds\int_{-\infty }^{\infty }dz\,G(\tau , y;s,z) \left[f._1(u_0(s,z), u_0^{\prime}(s,z))U_1(s,z)\right.\\
&\qquad\qquad\qquad\qquad\qquad\qquad\left.+ f._2(u_0(s,z) ,u_0^{\prime}(s,z))U_1^{\prime}(s,z)\right].
\end{split}
\end{equation}
The prime $^{\prime}$ denotes the derivative with respect to the second argument and $f._j$ is the derivative of the function $f$ with respect to the $j$th variable.
\end{theorem}
\begin{proof}
By means of the change of variables as $x = Ke^y,t = T- 2\tau/\sigma^2$ and
\begin{equation*}
\frac{\partial }{\partial t} = -\frac{\sigma^2}{2} \frac{\partial }{\partial \tau }, \ \ \frac{\partial }{\partial x} = \frac{1}{x} \frac{\partial }{\partial y},
\end{equation*}
the modified Black-Scholes PDE (\ref{NonLinearBS}) and the terminal condition $\Phi(T,X_T) = (X_T- K)^+$ are rewritten for the unknown function $v(\tau ,y):=K^{-1}\Phi(t,x)$ as
\begin{eqnarray*}
 \frac{\partial v(\tau ,y)}{\partial \tau } &=&  \frac{\partial^2 v(\tau ,y)}{\partial y^2}  - \frac{\partial v(\tau ,y)}{\partial y}  + \frac{2\rho K}{\sigma^2} \sqrt{v(\tau ,y)^2 + \left( \frac{\partial v(\tau ,y)}{\partial y}  \right) ^2 } \label{PDE2} \\
         v(0,y)  &=&  \max \{ e^y-1,0 \} .
\end{eqnarray*}
By introducing the new unknown function $u=u(\tau,y)$ defined as $v(\tau,y) = e^{\frac{y}{2}-\frac{1}{4}\tau} u(\tau ,y)$, we obtain the canonical form of diffusion equation
\begin{equation*}
\frac{\partial u}{\partial \tau } = \frac{\partial^2 u}{\partial y^2}  + \rho f\left(u(\tau,y),u^{\prime}(\tau,y)\right).
\end{equation*}
Here the terminal condition is changed to $u(0,y) = \max \{ e^{\frac{y}{2}}-e^{-\frac{y}{2}} ,0 \} $.
By introducing an unknown function $B(k,\tau)$, suppose that the solution of (\ref{PDE2}) has the form
\begin{eqnarray}
u(\tau, y) = u_0(\tau,y) + \int_{-\infty }^{\infty } \frac{1}{\sqrt{2\pi }} B(k,\tau ) e^{iky} dk,
\label{sol_form50}
\end{eqnarray}
where $u_0(\tau,y)$ is the solution for the case $ \rho =0$, i.e., $(\partial _{\tau } - \partial^2_y )u_0(\tau , y) =  0$. Thus,
\begin{equation*}
u_0(\tau,y) =  \int_{-\infty }^{\infty } G(\tau , y;0,z)  \max \{ e^{\frac{z}{2}}-e^{-\frac{z}{2}} ,0 \}   dz.
\end{equation*}
Inserting the representation of $u_0(\tau,y)$ into (\ref{PDE2}) yields
\begin{equation*}
(\partial_{\tau} - \partial^2_y )u  \ = \int_{-\infty }^{\infty } \frac{1}{\sqrt{2\pi }} \Bigl{\{ }  \frac{\partial B(k,\tau )}{\partial \tau} +k^2B(k,\tau )  \Bigl{ \} } e^{iky} \ = \ \rho f(u,u^{\prime}).
\end{equation*}
Via Fourier transform,
\begin{equation}
 \frac{\partial B(k,\tau )}{\partial \tau}   = -k^2  B(k,\tau ) +  \rho \tilde{f}( \tau ,k ),
 \label{ODE51}
\end{equation}
where
\begin{equation*}
\tilde{f}(\tau,k)  = \int_{-\infty }^{\infty } \frac{1}{\sqrt{2\pi }} f\bigl( u(\tau,y),u^{\prime}(\tau,y) \bigl)  e^{-iky}dy.
\end{equation*}
We solve (\ref{ODE51}) via variation of parameters. By introducing new function $\tilde{B}(k,\tau )$, we assume that the solution has the form
\begin{equation}
B(k,\tau ) = e^{ -k^2\tau }  \tilde{B}(k,\tau ).
\label{B_def}
\end{equation}
Inserting this into (\ref{ODE51}) gives
\begin{equation*}
 e^{ -k^2\tau } \ \frac{\partial \tilde{B}(k,\tau )}{\partial \tau } \ = \  \rho  \tilde{f}(\tau ,k),
\end{equation*}
which is equivalent to
\begin{eqnarray*}
\tilde{B}(k,\tau ) \ = \  \rho  \int_0^{\tau }  e^{k^2 t } \ \tilde{f}(t , k) dt.
\end{eqnarray*}
Consequently, the difference between the arbitrage solution $u$ and the no arbitrage solution $u_0$ is
\begin{equation}
\begin{split}
&u(\tau,y)-u_0(\tau,y)=\\
                      &= \int_{-\infty }^{\infty } \frac{1}{\sqrt{2\pi }} e^{iky} B(\tau,k)dk \\
                      &= \int_{-\infty }^{\infty }  \frac{1}{\sqrt{2\pi }} e^{iky}e^{ -k^2 \tau }  \Biggl{\{ }    \rho  \int_0^{\tau }  e^{k^2 t } \ \tilde{f}(t , k) dt   \Biggl{ \} } dk\\
                      &= \rho \frac{1}{2\pi }  \int_0^{\tau } \Biggl(       \int_{-\infty }^{\infty }  \Biggl{\{ }  \int_{-\infty }^{\infty } e^{ -k^2 (\tau -s) +ik(y-z) } \ f\bigl( u(s,z),u^{\prime}(s,z) \bigl)  dz  \Biggl{ \} } dk   \Biggl) ds \\
                      &= \rho \int_0^{\tau }  \Bigl( \int_{-\infty }^{\infty } G(\tau , y;s,z)f\bigl( u(s,z),u^{\prime}(s,z) \bigl) dz \Bigl) ds=:\rho F[u](\tau ,y).
\end{split}
\end{equation}
The non linear Black-Scholes PDE (\ref{NonLinearBS}) with the terminal condition is therefore equivalent to the functional equation
\begin{equation}\label{funeq}
G[u]:=u-u_0-\rho F[u]=0,
\end{equation}
which can be solved by a Newton's approximation scheme. The first element of the approximation sequence of the solution $u$ is $u_0$. The second, $u_1$ is the solution of the linearization of (\ref{funeq}) at $u_0$
\begin{equation}\label{linfuneq}
G[u_0]+G^*[u_0].(u_1-u_0)=0,
\end{equation}
where the star $\*$ denotes the G\^{a}teaux derivative. The solution reads
\begin{equation*}
\begin{split}
u_1&=u_0+\rho({\bf 1}-\rho F^*[u_0])^{-1}.F[u_0]\\
&=u_0+\rho({\bf 1}+\rho F^*[u_0]).F[u_0]+O(\rho^3)\\
&=u_0+\rho U_1+\rho^2 F^*[u_0].U_1+O(\rho^3)\quad(\rho\rightarrow0),
\end{split}
\end{equation*}
where $U_1:=F[u_0]$ corresponds to (\ref{defu1}). We now compute the G\^{a}teaux derivative of $F$ at $u_0$ as
\begin{equation*}
\begin{split}
&F^*[u].U_1(\tau ,y) \\
 &\quad=\int_0^{\tau}\,ds\int_{-\infty }^{\infty}dz\, G(\tau, y;s,z)\bigl[ f._{1}(u_0(s,z) , u_0^{\prime}(s,z))U_1(s,z) \\
 &\quad\quad\quad\quad\quad\quad\quad\quad\quad\quad\quad\quad\quad + f._{2}(u_0(s,z), u_0^{\prime}(s,z))U_1^{\prime}(s,z)\bigl].
\end{split}
\end{equation*}
We can now derive the second order approximate solution for $u$ as
\begin{equation*}
\begin{split}
&u(\tau,y)=u_0(\tau,y)+\rho \int_0^{\tau }ds\int_{-\infty }^{\infty }dz\,  G(\tau , y;s,z)f(u_0(s,z),u^{\prime}_0(s,z)) \\
          &\quad\quad\quad\quad\quad\quad + \rho ^2 \int_0^{\tau}\,ds\int_{-\infty }^{\infty}dz\, G(\tau, y;s,z)\bigl [u_1(s,z)f._{1}(u_0(s,z) , u_0^{\prime}(s,z)) \\
          &\quad\quad\quad\quad\quad\quad\quad\quad\quad + u_1^{\prime}(s,z)f._{2}(u_0(s,z), u_0^{\prime}(s,z))\bigl] \\
          & \quad\quad\quad\quad\quad\quad + O(\rho^3)\quad(\rho\rightarrow0).
\end{split}
\end{equation*}
By tracing back of the change of variables in (\ref{NonLinearBS}) we can obtain the solution $\Phi(t,X_t)$ as in (\ref{solcall}).\\
\end{proof}

%
%
%
%
\section{Conclusions}
We apply Geometric Arbitrage Theory to obtain results in Mathematical Finance, which do not need stochastic differential geometry in their formulation. First, we utilize the equivalence for a certain subclass of It\^{o} processes between the no-unbounded-profit-with-bounded-risk condition and the expected utility maximization to prove the equivalence between the (NUPBR) condition with the (ZC) condition. Then, we generalize the Black-Scholes PDE to markets allowing arbitrage, computing an approximated solution for the non linear PDE for a call option with arbitrage.

\section*{Acknowledgements}

We are grateful to the participants of the Quantitative Methods in Finance congress (Sydney, December 2017) and 10th World Congress of Bachelier Finance Society (Dublin, July 2018) for valuable discussions, especially for Stefan Tappe, suggesting the relation between (NUPBR) and (ZC).
We would like to extend our gratitude to Claudio Fontana, who highlighted that for a previous incorrect version of Proposition \ref{NovikovThm} Bessel's processes, which satisfy (NUPBR) but not (NFLVR) as shown in \cite{Fo15} and in \cite{FoRu13}, would have been a counterexample, thus leading to the current corrected version.

\appendix
\section{Generalized Derivatives of Stochastic Processes}\label{Derivatives}
In stochastic differential geometry one would like to lift
the constructions of stochastic analysis from open subsets of
$\mathbf{R}^N$ to  $N$ dimensional differentiable manifolds. To that
aim, chart invariant definitions are needed and hence a stochastic
calculus satisfying the usual chain rule and not It\^{o}'s Lemma is required,
(cf. \cite{HaTh94}, Chapter 7, and the remark in Chapter 4 at the
beginning of page 200). That is why the papers about geometric arbitrage theory are mainly concerned in
 by stochastic integrals and derivatives meant in \textit{Stratonovich}'s
sense and not in \textit{It\^{o}}'s. Of course, at the end of the computation, Stratonovich integrals can be transformed into It\^{o}'s.
Note that a fundamental portfolio equation, the self-financing condition cannot be directly formally expressed with Stratonovich integrals, but first with It\^{o}'s and then transformed into Stratonovich's, because it is a non-anticipative condition.
\begin{definition}\label{Nelson}
Let $I$ be a real interval and $Q=(Q_t)_{t\in I}$ be a  $\mathbb{R}^N$-valued stochastic process on the probability space
$(\Omega, \mathcal{A}, \mathbb{P})$. The process $Q$ determines three families of $\sigma$-subalgebras of the $\sigma$-algebra $\mathcal{A}$:
\begin{itemize}
\item[(i)] ``Past'' $\mathcal{P}_t$, generated by the preimages of Borel sets in $\mathbf{R}^N$  by all mappings $Q_s:\Omega\rightarrow\mathbf{R}^N$ for $0<s<t$.
\item[(ii)] ``Future'' $\mathcal{F}_t$, generated by the preimages of Borel sets in $\mathbf{R}^N$  by all mappings $Q_s:\Omega\rightarrow\mathbf{R}^N$ for $0<t<s$.
\item[(iii)] ``Present'' $\mathcal{N}_t$, generated by the preimages of Borel sets in $\mathbf{R}^N$  by the mapping $Q_s:\Omega\rightarrow\mathbf{R}^N$.
\end{itemize}
Let $Q=(Q_t)_{t\in I}$ be continuous.
Assuming that the following limits exist, \textbf{Nelson's stochastic derivatives} are defined as
\begin{equation}
\begin{split}
&\mathfrak{D}Q_t:=\lim_{h\rightarrow 0^+} \mathbb{E} \Bigl[  \frac{Q_{t+h}-Q_t}{h} \Bigl| \mathcal{P}_t \Bigl]  \text{: forward derivative,}\\
& \mathfrak{D}_*Q_t:=\lim_{h\rightarrow 0^+}\mathbb{E}\Bigl[ \frac{Q_{t}-Q_{t-h}}{h} \Bigl| \mathcal{F}_{t} \Bigl]  \text{: backward derivative,}\\
&\mathcal{D}Q_t:=\frac{\mathfrak{D}Q_t+\mathfrak{D}_*Q_t}{2} \text{: mean derivative}.
\end{split}
\end{equation}

Let $\mathcal{S}^1(I)$ the set of all processes $Q$ such that
$t\mapsto Q_t$, $t\mapsto \mathfrak{D}Q_t$ and $t\mapsto \mathfrak{D}_*Q_t$ are continuous
mappings from $I$ to $L^2(\Omega, \mathcal{A})$. Let
$\mathcal{C}^1(I)$ the completion of $\mathcal{S}^1(I)$ with respect
to the norm
\begin{equation}
\|Q\|:=\sup_{t\in I} \Bigl(   \|Q_t\|_{L^2(\Omega, \mathcal{A})}+\|\mathfrak{D}Q_t\|_{L^2(\Omega, \mathcal{A})}+\|\mathfrak{D}_*Q_t\|_{L^2(\Omega, \mathcal{A})} \Bigl).
\end{equation}
\end{definition}
\begin{remark}
The stochastic derivatives $\mathfrak{D}$, $\mathfrak{D}_*$ and  $\mathcal{D}$
correspond to It\^{o}'s, to the anticipative and, respectively,  to Stratonovich's integral (cf. \cite{Gl11}). The process space $\mathcal{C}^1(I)$ contains all It\^{o} processes. If $Q$ is a Markov process, then the sigma algebras $\mathcal{P}_t$ (``past'') and $\mathcal{F}_t$ (``future'') in the definitions of forward and backward derivatives can be substituted by the sigma algebra $\mathcal{N}_t$ (``present''), see Chapter 6.1 and 8.1 in (\cite{Gl11}).
\end{remark}

\noindent Stochastic derivatives can be defined pointwise in $\omega\in\Omega$ outside the class $\mathcal{C}^1$ in terms of generalized functions.
\begin{definition}
Let $Q:I\times\Omega\rightarrow\mathbb{R}^N$ be a continuous linear functional in the test processes $\varphi:I\times\Omega\rightarrow\mathbb{R}^N$ for $\varphi(\cdot,\omega)\in C^{\infty}_c(I,\mathbb{R}^N)$. We mean by this that for a fixed $\omega\in\Omega$ the functional $Q(\cdot,\omega)\in\mathcal{D}(I,\mathbb{R}^N)$, the topological vector space of continuous distributions. We can then define
\textbf{Nelson's generalized stochastic derivatives:}
\begin{equation}
\begin{split}
&\mathfrak{D}Q(\varphi_t):=-Q(\mathfrak{D}\varphi_t)\text{: forward generalized derivative,}\\
& \mathfrak{D}_*Q(\varphi_t):=-Q(\mathfrak{D}_*\varphi_t)\text{: backward generalized derivative,}\\
&\mathcal{D}  Q  (\varphi_t):=-Q(\mathcal{D}\varphi_t)\text{: mean generalized derivative}.
\end{split}
\end{equation}
\end{definition}
\noindent If the generalized derivative is regular, then the process has a derivative in the classic sense. This construction is nothing else than a straightforward pathwise lift of the theory of generalized functions to a wider class of stochastic processes which do not a priori allow for Nelson's derivatives in the strong sense.

\bibliographystyle{siamplain}
\bibliography{references}

\begin{thebibliography}{99}

\bibitem{Be01}
  Becherer, D., The Num\'{e}raire Portfolio for Unbounded Semimartingales.
  {\itshape Finance and Stochastics}, 2001, {\bfseries 5}, 327--341.

\bibitem{BeFr02}
  Bellini, F. and Frittelli, M.,  On the Existence of Minimax Martingale Measures.
  {\itshape Mathematical Finance},2002, {\bfseries 12}, 1--21.

\bibitem{BjHu05}
  Bj\"ork,T. and Hult, H.,  A Note on Wick Products and the Fractional Black-Scholes Model.
  {\itshape Finance and Stochastics}, 2005, {\bfseries 9}, 197--209.

\bibitem{Bl81}
  Bleecker, D.,
  {\itshape Gauge Theory and Variational Principles}.
  Addison-Wesley Publishing. (1981) (republished by Dover 2005).

\bibitem{ChLa07}
  Christensen, M. M. and Larsen, K., No Arbitrage and Growth Optimal Portfolio.
  {\itshape Stochastic Analysis and Applications}, 2007, {\bfseries 25},  255--280.

\bibitem{DeSc94}
  Delbaen, F. and Schachermayer, W., A General Version of the Fundamental Theorem of Asset Pricing.
  {\itshape Mathematische Annalen}, 1994, {\bfseries 300}, 463--520.

\bibitem{DeSc08}
  Delbaen, F. and Schachermayer, W.,
  {\itshape The Mathematics of Arbitrage}.
  Springer-Verlag Berlin Heidelberg, 2008.

\bibitem{DeMe80}
  Dellach\'{e}rie, C., and Meyer, P. A.,
  {\itshape Probabilit\'{e} et potentiel II - Th\'{e}orie des martingales - Chapitres 5 \`{a} 8},
  Hermann, 1980.

\bibitem{El82}
  Elworthy, K. D.,
  {\itshape Stochastic Differential Equations on Manifolds},
  London Mathematical Society Lecture Notes Series, 1982.

\bibitem{Em89}
  Em\'{e}ry, M.,
  {\itshape Stochastic Calculus on Manifolds-With an Appendix by P. A. Meyer},
  Springer-Verlag Berlin Heidelberg, 1989.

\bibitem{FaVa12}
  Farinelli, S. and Vazquez, S., Gauge Invariance, Geometry and Arbitrage.
  {\itshape The Journal of Investment Strategies}, 2012, {\bfseries 1}, 23--66.

\bibitem{Fa15}
  Farinelli, S., Geometric Arbitrage Theory and Market Dynamics.
  {\itshape Journal of Geometric Mechanics}, 2015, {\bfseries 7}, 431--471.

\bibitem{Fa21}
  Farinelli, S., Geometric Arbitrage Theory and Market Dynamics reloaded.
  preprint, {\rm arXiv}, 2021.

\bibitem{FaTa20}
   Farinelli, S., and Takada, H., Can You Hear the Shape of Market? Geometric Arbitrage and Spectral Theory,
  preprint, {\rm arXiv} 1509.03264.

\bibitem{FlHu96}
  Flesaker, B. and Hughston, L., Positive Interest.
  {\itshape Risk Magazine}, 1996, {\bfseries 9}, 46--49.

\bibitem{FoeSc04}
  F\"ollmer, H. and Schied, A.,
  {\itshape Stochastic Finance: An Introduction In Discrete Time}.
  Second Edition, De Gruyter Studies in Mathematics, 2004.

\bibitem{Fo15}
  Fontana, C., Weak and Strong No-Arbitrage Conditions for Continuous Financial Markets.
  {\itshape International Journal of Theoretical and Applied Finance}, 2015, {\bfseries 18}, 1--34.

\bibitem{FoRu13}
  Fontana, C. and Runggaldier, W. J., Diffusion-Based Models for Financial Markets without Martingale Measures.
  Chapter 4 In {\it Risk Measures and Attitudes} (eds. Biagini, F., Richter, A. and Schlesinger, H.),
  Springer-Verlag, London, 2013, 45--91.

\bibitem{Gl11}
  Gliklikh, Y. E.,
  {\itshape Global and Stochastic Analysis with Applications to Mathematical Physics},
  Theoretical and Mathematical Physics, Springer-Verlag London, 2010.

\bibitem{HaTh94}
 Hackenbroch, W. and Thalmaier, A.,
 {\itshape Stochastische Analysis. Eine Einf\"{u}hrung in die Theorie der stetigen Semimartingale},
 Teubner Verlag, 1994.

\bibitem{HePl06}
  Heath, D. and Platen, E.,  {\itshape A Benchmark Approach to Quantitative Finance}, Springer, 2006.

\bibitem{Ho03}
  H\"ormander, L.,
  {\itshape The Analysis of Linear Partial Differential Operators I: Distribution Theory and Fourier Analysis}.
  Springer-Verlag Berlin Heidelberg, 2003.

\bibitem{HuPr15}
  Hugonnier, J. and Prieto, R., Asset Pricing with Arbitrage Activity.
  {\itshape  Journal of Financial Economics}, 2015, {\bfseries 115}, 411-428.

\bibitem{Hs02}
   Hsu, E. P.,
  {\itshape Stochastic Analysis on Manifolds},
  Graduate Studies in Mathematics, 2002, {\bfseries 38}, AMS.

 \bibitem{HuSc10}
  Hulley, H. and Schweizer, M.,  $M^6$ - On Minimal Market Models and Minimal Martingale Measures.
  In {\itshape Contemporary Quantitative Finance. Essays in Honour of Eckhard Platen} (eds. Chiarella, C. and Novikov, A.), Springer-Verlag Berlin Heidelberg, 2010, 35--51.

 \bibitem{HuKe04}
  Hunt, P. J. and Kennedy, J. E.,
  {\itshape Financial Derivatives in Theory and Practice},
  Wiley Series in Probability and Statistics, John Wiley \& Sons, 2004.

\bibitem{Il00}
  Ilinski, K., Gauge Geometry of Financial Markets.
  {\itshape J. Phys. A: Math. Gen}, 2000, {\bfseries 33}, 5--14.

\bibitem{Il01}
Ilinski, K.,
  {\itshape Physics of Finance: Gauge Modelling in Non-Equilibrium Pricing}, Wiley, 2001.

\bibitem{JeYo79}
  Jeulin, T. and Yor, M.,  {\itshape In\'{e}galit\'{e} de Hardy, semimartingales, et faux-amis},
  S\'{e}minaire de Probabilit\'{e}s, XIII, 1979,  332--359.

\bibitem{Ka97}
  Kabanov, Y. M., On the FTAP of Kreps-Delbaen-Schachermayer.
  In {\itshape Statistics and Control of Stochastic Processes, The Liptser Festschrift Proceedings of Steklov Mathematical Institute Seminar, Moscow, Russia} (eds. Kabanov, Y. M.), World Scientific, Singapore, 1997, 191--203.

\bibitem{KaKa07}
  Karatzas, I. and Kardaras, C., The Num\'{e}raire Portfolio in Semimartingale Financial Models.
  {\itshape Finance and Stochastics}, 2007, {\bfseries 11}, 447--493.

\bibitem{KaKr94}
  Kabanov, Y. M. and Kramkov, D. O.,  Large Financial Markets: Asymptotic Arbitrage and Contiguity.
  {\itshape Probab. Theory and Its Applications},  {\bfseries 39}, 1994, 222--229.

\bibitem{LoWi00}
  Loewenstein, M. and Willard, G.A.,  Local martingales, arbitrage, and viability.
  {\itshape Economic Theory}, 2000,  {\bfseries 16}, 135--161.

\bibitem{Lu69} Luenberger David G., {\itshape Local Theory of Constrained Optimization: Optimization by Vector Space Methods}, New York John Wiley \& Sons, 1969.

\bibitem{Ma96}
  Malaney, P. N.,
  The Index Number Problem: A Differential Geometric Approach.
  Ph.D thesis, Harvard University Economics Department, 1996.

\bibitem{Pr10}
  Protter, Ph. E.,
  {\itshape Stochastic Integration and Differential Equations: Version 2.1}.
  Stochastic Modelling and Applied Probability, Springer, 2010.

\bibitem{Ro94}
  Rogers, L. C. G., Equivalent Martingale Measures and No-Arbitrage.
  {\itshape Stochastics, Stochastics Rep.}, 1994, {\bfseries 51},  41--49.

\bibitem{Ru13}
 Ruf, J.,  Hedging under Arbitrage.
 {\itshape  Mathematical Finance}, 2013, {\bfseries  23}, 297--317.

\bibitem{Scha01}
  Schachermayer, W.,  Optimal Investment in Incomplete Markets When Wealth May Become Negative.
  {\itshape Annals of Applied Probability}, 2001, {\bfseries 11}, 694--734.

\bibitem{Schw80}
  Schwartz, L.,
  {\itshape Semi-martingales sur des vari\'{e}t\'{e}s et martingales conformes sur des vari\'{e}t\'{e}s analytiques complexes}.
  Springer Lecture Notes in Mathematics, Springer-Verlag Berlin Heidelberg, 1980.

\bibitem{SmSp98}
  Smith, A. and Speed, C.
  {\itshape Gauge Transforms in Stochastic Investment}.
  Proceedings of the 1998 AFIR Colloquim, Cambridge, England. 1998.

\bibitem{St00}
  Stroock, D. W.,
  {\itshape An Introduction to the Analysis of Paths on a Riemannian Manifold}.
  Mathematical Surveys and Monographs, 2000, {\bfseries 74}, AMS.

\bibitem{We06}
  Weinstein, E.,
  {\itshape Gauge Theory and Inflation: Enlarging the Wu-Yang Dictionary to a unifying Rosetta Stone for Geometry in Application}.
  Talk given at Perimeter Institute, (2006).

\bibitem{Yo99}
  Young, K., Foreign Exchange Market as a Lattice Gauge Theory.
  {\itshape Am. J. Phys.}, 1999, {\bfseries 67},  862--868.

\bibitem{Ze95} Zeidler E., {\itshape Applied functional analysis: Variational Methods and Optimization}, Applied Mathematical Sciences 109, New York, NY, Springer-Verlag, 1995.


\end{thebibliography}

\end{document}